\documentclass[11pt]{article}

\usepackage{amsmath, amssymb, amsthm}
\usepackage{graphicx}
\usepackage{algorithm}
\usepackage{algpseudocode}
\newcommand{\ceil}[1]{\left\lceil #1 \right\rceil}
\newcommand{\imp}{\Rightarrow}

\theoremstyle{plain}
\newtheorem{theorem}{Theorem}[section]
\newtheorem{lemma}[theorem]{Lemma}
\newtheorem{proposition}[theorem]{Proposition}
\newtheorem{corollary}[theorem]{Corollary}
\theoremstyle{definition}
\newtheorem{definition}[theorem]{Definition}

\title{Pebbling Arguments for Tree Evaluation}
\author{David Liu\thanks{Supported by an Ontario Graduate Scholarship and NSERC CGS M scholarship.}}

\begin{document}

\maketitle
\begin{abstract}
The \emph{Tree Evaluation Problem} was introduced by Cook et al. in 2010 as a candidate for separating \textbf{P} from \textbf{L} and \textbf{NL} \cite{cook12}. The most general space lower bounds known for the Tree Evaluation Problem require a semantic restriction on the branching programs and use a connection to well-known pebble games to generate a bottleneck argument. These bounds are met by corresponding upper bounds generated by natural implementations of optimal pebbling algorithms. In this paper we extend these ideas to a variety of restricted families of both deterministic and non-deterministic branching programs, proving tight lower bounds under these restricted models. We also survey and unify known lower bounds in our ``pebbling argument'' framework.
\end{abstract}

\section{Introduction}
Complexity theory is the study of the hardness of problems. Starting from the measurement of the classical resources Turing machine time and space, the research over the past fifty years has seen the exciting development of new ideas to power and analyse algorithms: randomness, communication, circuits, and quantum computing models are among the most famous of these. This has led to a proliferation of complexity classes, and it is no exaggeration to say that the most important open questions in computer science concern the exact relationship between them, the classic ``\textbf{P} = \textbf{NP}?" being the most famous of these. These questions are inherently difficult to answer and have resisted proofs for decades, though the past few years have seen some large successes with the collapse of \textbf{SL} into \textbf{L} \cite{rein08}, that \textbf{QIP} = \textbf{PSPACE} \cite{jjuw10}, and that \textbf{ACC} $\subsetneq$ \textbf{NEXP} \cite{will11}.

In this paper, we explore an avenue of attack to separate \textbf{L} and \textbf{NL} from \textbf{P} studied by Cook et al. in \cite{cook12}. They introduced the \emph{Tree Evaluation Problem} as a candidate to separate these complexity classes, studying space lower bounds using branching programs as their model of computation. Of course, proving good general lower bounds is most likely a very  difficult problem; on the other hand, applying various restrictions to the branching program model makes this problem much more tractable.

There are natural algorithms solving this problem in both the deterministic and non-deterministic settings which implement optimal strategies for well-known pebbling games. Drawing inspiration from these algorithms, Cook et al. introduced the semantic \emph{thrifty} restriction on branching programs, and proved that the algorithm is optimal for deterministic thrifty branching programs. They conjectured this algorithm is in fact optimal for \emph{all} deterministic branching programs; proving this conjecture would separate \textbf{L} from \textbf{P}.
Proving the analogous conjecture for non-deterministic branching programs would separate \textbf{NL} and \textbf{P}. However, the algorithm is not known to be optimal even under the non-deterministic thrifty setting.

\subsection{Our Contributions}
Our main contribution is to apply the pebbling arguments to prove lower bounds for various families of restricted branching programs. For each of these restrictions, we derive a tight asymptotic lower bound corresponding precisely to the pebble number of the corresponding pebbling games. Our most significant contribution in this vein is our deterministic read-once lower bound.
\begin{theorem}
\label{thm:ro_det}
Any deterministic read-once branching program solving the Tree Evaluation Problem has at least $k^h$ states.
\end{theorem}
To prove this theorem, we create a variant on the pebbling game to extend the metaphor beyond just the \emph{correct} values of the nodes to include also non-thrifty function values. We have an analogous theorem for restricted non-determinstic branching programs.
\begin{theorem}
\label{thm:ro_ndet}
Let $B$ be a non-deterministic thrifty branching program solving $TEP^h_2(k)$. If $B$ is syntactic read-once, null-path-free, or semantic read-once, then $B$ has at least $\displaystyle{k^{\ceil{\frac{h}{2}}+1}}$ states.
\end{theorem}
By presenting all of our proofs in a similar structure, we establish a general framework that may be useful in generalizing the pebbling argument to broader, non-thrifty classes of branching programs.

Our second contribution is to simplify two lower bounds already in the literature and unify them by fitting them to our framework. First, we take the deterministic thrifty lower bound in \cite{cook12} and simplify it by making more explicit use of the pebbling argument and a straightforward \emph{tag} argument found in a few other proofs in that paper. We then present a simplified version of the thrifty bitwise-independent lower bound in \cite{ks13}, refining the ideas in this paper and introducing the more natural restriction of \emph{node}-independence. Specifically, we prove the following new theorem.

\begin{theorem}
\label{thm:niro}
Every non-deterministic, node-independent, read-once branching program $B$ solving $TEP^h_2(k)$ has at least $k^{\frac{h}{2} + 1}$ states.
\end{theorem}

The ideas and extension we present related to the bitwise- and node-independent restrictions may also have a hope of generalizing to larger classes of branching programs. 

\subsection{Related Work}
Using black pebbles as a device capturing deterministic space dates back to the 70s \cite{sethi75}, and white pebbles are a natural generalization to non-deterministic space. Whole pebbles turned out to be too restrictive a model, and so a black-white \emph{fractional} pebbling game was introduced in \cite{cook12} for non-deterministic branching programs. Their lower bounds come in two flavours: those for restricted branching programs and arbitrary height trees using pebbling arguments, and those for unrestricted branching programs but small height trees using more ad hoc methods. 

Recently, Komarath and Sarma developed a new restriction called \emph{bitwise-independence} and successfully applied a pebbling argument to non-deterministic thrifty branching programs with this restriction. In an unpublished work, Siu Man Chan and James Cook derived a tight lower bound for deterministic read-once branching programs using polynomials over finite fields. The proof that we give in this paper uses a different strategy that we hope stands a better chance of generalizing. 

Some work has been done in the more general DAG Evaluation Problem, where the underlying graph is an arbitrary DAG rather than a complete binary tree. Wehr \cite{wehr11} proved an analogous lower bound for deterministic thrifty branching programs solving DAG Evaluation. Chan \cite{chan13}, using different pebble games, studied circuit depth lower bounds for DAG Evaluation under a semantic restriction called output-relevance, closely related to thriftiness. Because of the more general nature of DAG Evaluation, Chan achieved a separation of \textbf{NC}$^i$ and \textbf{NC}$^{i+1}$ for each $i$, as well as separating \textbf{NC} and \textbf{P}, under this semantic restriction.

\subsection{Organization of the Paper}
The remainder of the paper is structured as follows. Section 2 outlines the notation and problem context for the rest of the paper. In Section 3, we present a simplified proof of the deterministic thrifty lower bound found in \cite{cook12}; this serves as an introduction to the type of pebbling argument that will become more elaborate in the rest of the paper. Sections 4 and 5 are the main results for non-deterministic read-once thrifty and deterministic read-once branching programs, respectively. Section 6 is an exposition of the lower bound for bitwise-independence, and Section 7 refines these ideas into a new lower bound for read-once branching programs. Finally, in Section 8 we conclude the paper with a discussion of the pebbling argument and promising directions for future research.

\section{Preliminaries}
\subsection{Tree Evaluation}
Let $T^h_2$ be the full binary tree with $h$ levels, so that the total number of nodes in $T^h_2$ is $2^h - 1$. We number the nodes in the standard heap order, so that the root is 1 and the children of node $i$ are labelled $2i$ and $2i + 1$. 

\begin{definition}[Tree Evaluation Problem]
The \textbf{Tree Evaluation Problem} on $T^h_2$ and positive integer $k$ takes as input the following ($[k]$ denotes the set $\{1,\dots, k\}$):
\begin{enumerate}
\item[(i)] For each leaf $i$ of $T^h_2$, a number $v_i \in [k]$;
\item[(ii)] For each internal node $i$, a function $f_i: [k] \times [k] \to [k]$, given as a function table of $k^2$ entries from $[k]$.
\end{enumerate}
Thus the input consists of $2^{h-1}k + (2^{h-1}-1)k^2$ values from $[k]$.

The goal is to ``evaluate the tree" in the natural way. The notation $v_i$ has the semantic meaning of the \emph{correct} value of node $i$; thus the correct values of the leaves are already specified by the input. For every internal node $i$, we define $v_i = f_i(v_{2i}, v_{2i + 1})$ inductively, and the goal is to find $v_1$.
\end{definition}

We use the notation $TEP^h_2(k)$ to denote this problem, and to emphasize that we are interested in $k$ as the argument of interest. 

\subsection{Branching Programs}
Though separation of \textbf{P} from \textbf{L} requires only uniform lower bounds, analysis of branching programs, a non-uniform model of sequential computation, has proven more tractable (at least at present). Though these lower bounds are stronger than their uniform versions, it remains unclear how to take advantage of the uniformity of Turing machines to prove good lower bounds for this problem.

We remark that many variations of branching programs have been studied in literature; for this paper, we use the model from \cite{cook12}, which we define now.

\begin{definition}[Branching Program]
A $k$-way \textbf{branching program} (BP) is a directed graph with labels on both the nodes (called states) and the edges. There are $k$ sink states, each labelled with a distinct number from $[k]$, called \emph{output} states. 
Every other state is labelled with a query to a particular input value; for $TEP^h_2(k)$, this is either a leaf value $v_i$ or some internal function query $f_i(x,y)$ for some $x,y \in [k]$. 
Each edge is labelled with a number in $[k]$; these are interpreted as possible ``results" of the query made by the edge's tail state.
Finally, we assume that there is only one source state, and this is distinguished as the \emph{start} state.
\end{definition}

Each problem input generates \emph{computations} on the branching program, which are paths beginning at the start state whose edges are consistent with the input instance according to the interpretation of the edges and states given by the preceding definition.
A \emph{complete} path is a path from the start state to an output state; not all computation paths are complete. Conversely, not every complete path can be followed by an input; these are called \emph{null-paths} and contain two states which query the same variable but take two differently labelled edges out of them.

A branching program is \emph{deterministic} if every non-output state has exactly $k$ out-edges, each labelled with a distinct number in $[k]$. Otherwise, it is \emph{non-deterministic}. On deterministic branching programs, each input induces a unique, complete, computation path. A branching program \emph{computes} a function $f$ if for each input $I$ to $f$, at least one computation path induced by $I$ must be complete, and \emph{every} complete computation path induced by $I$ ends at the output state labelled $f(I)$. Note that if we require just one complete path to give the right answer, the model is trivial for non-deterministic branching programs. For non-deterministic branching programs, we will generally identify each input with a single induced complete computation path (arbitrarily chosen), referring to it as $C(I)$. Note that this notation carries over to the deterministic case, where there is just one choice for each $C(I)$.

If $\gamma$ is a state on $C(I)$, we use $C_0(I, \gamma)$ to denote the segment of $C(I)$ before $\gamma$, but including the edge leading into $\gamma$; we use $C_1(I, \gamma)$ to denote the segment of $C(I)$ after and including $\gamma$. Inputs $I$ and $J$ \emph{agree before $\gamma$} if $C_0(I, \gamma)$ and $C_0(J, \gamma)$ are identical paths (states and edges).
Agreement after $\gamma$ is defined similarly. 

We measure the \emph{size} of the branching program as its number of states, which is exponentially related to the corresponding Turing Machine space. More specifically, Cook et al. \cite{cook12} showed that to prove a (non-)deterministic super-logarithmic lower bound on the space complexity of $TEP^h_2(k)$, it suffices to prove an asymptotic lower bound on the (non-)deterministic branching program size of $\Omega(k^{g(h)})$, where the $\Omega$ is with respect to $k$, and $g(h)$ is an \emph{unbounded} function depending only on $h$.

Finally, we introduce two main branching program restrictions studied in this paper.

\begin{definition}[Thrifty]
A computation path $C(I)$ is \emph{thrifty} if for all internal node queries $f_i(x,y)$ made on $C(I)$, $x = v_{2i}^I$ and $y = v_{2i+1}^I$. A branching program is \emph{thrifty} if each complete computation path on $C(I)$ is thrifty.
\end{definition}

\begin{definition}[Read-Once]
A branching program is \emph{syntactic read-once} if every complete path queries every value of the input at most once.
A branching program is \emph{semantic read-once} if this restriction holds for every complete computation path. Every syntactic read-once branching program is also semantic read-once, but the converse is certainly not true.
\end{definition}

\subsection{Pebbling}
We use two main pebbling games in this paper. The simplest version is the \emph{whole black pebble game}, which can be described as follows. In this game, a sequence of pebble configurations (i.e., pebbles on nodes) is \emph{valid} if the first configuration is empty, the last configuration has just a single pebble located on the root node, and each configuration is transformed into the next by applying one of the following moves:
\begin{itemize}
\item Place a black pebble on a leaf.
\item If $i$ is an internal node and all of its children are pebbled, place a black pebble on $i$ and simultaneously remove pebbles from all, some, or none of its children. This is known as a \emph{black sliding move}.
\item Remove a pebble.
\end{itemize}
The goal is to find a valid pebbling sequence that uses the \emph{fewest} number of pebbles, where the number of pebbles used by a sequence is the maximum number of pebbles on any one configuration in the sequence. For the complete binary tree $T^h_2$, the following lower bound is known (for example, \cite{cook12}).

\begin{theorem}[Folklore]
\label{thm:black_num}
Every valid black pebbling sequence of $T^h_2$ contains a configuration with at least $h$ pebbles. Moreover, this is tight; there exists a valid pebbling of $T^h_2$ using only $h$ pebbles.
\end{theorem}

The connection between this game and the Tree Evaluation Problem is to interpret pebbles as marking the nodes for which the program ``knows" the correct values at a point in the computation. The maximum number of pebbles used in a sequence then corresponds to the maximum amount of ``memory" used during the computation. In fact, implementing a minimal pebbling sequence as a branching program in the natural way yields the smallest known deterministic branching programs solving the Tree Evaluation Problem, and this is conjectured to be optimal.
\begin{corollary}
\label{cor:black_opt}
There is a deterministic branching program solving $TEP^h_2(k)$ that contains $\Theta(k^h)$ states.
\end{corollary}

The \emph{fractional black-white pebble game} introduces both white pebbles and fractional pebble values, which respectively capture the notions of non-deterministic guesses and partially known/guessed values. For each node $i$, a pebbling configuration stores the black and white pebble values $b(i)$ and $w(i)$. These change according to the following rules, subject to the conditions $0, \leq b(i), w(i) \leq 1$ and $b(i) + w(i) \leq 1$.
\begin{itemize}
\item Increase $w(i)$ or decrease $b(i)$ for some node $i$.
\item Increase $b(i)$ or decrease $w(i)$ for some \emph{leaf} $i$.
\item If $i$ is an internal node and all of its children are fully pebbled (i.e., $b(j) + w(j) = 1$), then increase $b(i)$ or decrease $w(i)$. If $b(i)$ increases, simultaneous decrease $b(j)$ for any children $j$ of $i$.
\end{itemize}
The goal for this game is to find a sequence of valid moves that begin and end with empty pebble configurations, and has a configuration where the root has a full black pebble. The \emph{whole black-white pebble game} has the additional restriction that $b(i),w(i) \in \{0,1\}$; that is, only whole black and white pebbles can be placed/removed. While the lower bound for whole black-white pebbling $T^h_2$ was also derived in \cite{cook12}, it was Vanderzwet who proved a tight lower bound on the pebble number of the corresponding fractional game \cite{van12}.

\begin{theorem}[\cite{cook12}, \cite{van12}]
\label{thm:frac_num}
Every valid black-white whole pebbling sequence of $T^h_2$ contains a configuration with at least $\ceil{\frac{h}{2}} + 1$ pebbles. If fractional pebbles are allowed, at least $\frac{h}{2} + 1$ pebbles are required. Both these bounds are tight.
\end{theorem}

\begin{corollary}
\label{cor:frac_opt}
There is a non-deterministic branching program solving $TEP^h_2(k)$ that contains $\Theta(k^{\frac{h}{2} + 1})$ states.
\end{corollary}

\section{Deterministic Thrifty}

As a warm-up, we present the proof of the lower bound for deterministic thrifty branching programs in \cite{cook12},
using this opportunity to illustrate the pebbling argument built upon in future sections. 
For some intuition behind this argument, consider an input $I$ to the Tree Evaluation Problem. We can view the states on $C(I)$ as storing information about the input, with the labeled edges between consecutive states acting as the mechanism of learning new information. Thus at the start state no information is known about the input, while at the output states precisely the correct value of the root is known. In general, the information learned along a computation path can be very complex, and non-determinism enables computation paths to ``guess" even more. However, for certain families of restricted branching programs, information about inputs can be learned and guessed only in very structured ways, and thus modeled by pebbling games. We adopt the following strategy to capitalize on this relationship:
\begin{enumerate}
\item[(1)] Given an input and induced complete computation path, associate the states on the path to configurations in a pebbling sequence.
\item[(2)] Argue that the pebbling sequence is \emph{valid} for some pebbling game.
\item[(3)] Argue that the pebbles associated with a state represent information about the input encoded at that state.
\item[(4)] Apply a pebbling lower bound to argue that each computation path has a \emph{supercritical state} which ``knows" a lot about the input. This state acts as a bottleneck for the inputs.
\end{enumerate}

\subsection{Pebbling Sequence}
Recall that in the deterministic setting, every input generates a unique computation path, which is always complete. We use the following proposition to associate pebbling sequences to these computation paths.
\begin{proposition}
\label{prop:td_crit}
Let $C(I)$ be a complete computation path on a deterministic thrifty branching program solving $TEP^h_2(k)$. Then every node of $T^h_2$ is queried on $C(I)$ at least once. Moreover, each non-root node is queried at least once before its parent is queried.
\end{proposition}
\begin{proof}
If $C(I)$ doesn't query the root, then $B$ makes a mistake on the input $I'$ which is identical to $I$ but has a different value at the thrifty root query $f_1(v^I_2, v^I_3)$, as $I'$ would follow $C(I)$ to an incorrect output state. Now suppose there exists a non-root node $i$ with parent $j$ such that $j$ is queried on $C(I)$ at some state $\gamma$, but $i$ is not queried on $C_0(I, \gamma)$. Let $I'$ be some input which differs from $I$ only on the value of the thrifty query to node $i$; then $I'$ agrees with $I$ before $\gamma$, and hence $C(I')$ includes $\gamma$. But then $\gamma$ makes a non-thrifty query with respect to $I'$.
\end{proof}

Proposition \ref{prop:td_crit} establishes that we may define the following special states along a computation path.
\begin{definition}[Critical State]
The \textbf{critical states} on $C(I)$ are defined recursively as follows:
\begin{itemize}
\item The critical state of the root is the last state on $C(I)$ that queries the root.
\item The critical state of a non-root node is the last state that queries it before the critical state of its parent.
\end{itemize}
\end{definition}
We now assign a \emph{black} pebbling in the obvious way, performing one pebbling move at each critical state:
\begin{itemize}
\item At the critical state of leaf, put a pebble on the leaf.
\item At the critical state of an internal node, put a black pebble on the node and simultaneously remove all pebbles from its children. Note that this is done in a single black sliding move.
\end{itemize}
The first critical state is associated with the empty configuration. The configuration produced as the result of a pebbling move at a critical state is associated with the \emph{next} critical state, with the exception of the final configuration with a pebble on the root, which is associated with the output state that ends $C(I)$. It follows immediately from the definition of the critical states that the resulting sequence is a valid black pebbling of $T^h_2$.

\subsection{A Tag Argument}
By Theorem \ref{thm:black_num}, each computation path $C(I)$ has a \emph{supercritical state}, a critical state whose associated pebbling configuration contains $h$ pebbles. 
Intuitively, one can use thriftiness to recover the correct values of the $h$ pebbled nodes from $C_1(I, \gamma)$.
We will say that the value of node $i$ is \emph{learned} on $C_1(I, \gamma)$ if the parent of $i$ is queried before $i$ is queried (or if $i$ is never queried). If the query to the parent has argument $a$ for node $i$, then we say that node $i$ is \emph{learned to have value} $a$; by thriftiness, $v^I_i = a$. 
Every pebbled node at $\gamma$ has its value learned on $C_1(I,\gamma)$.

We borrow the language of the proof of Theorem 5.15 in \cite{cook12} and define a tagging function on the set of inputs as follows: $U(I) = (\gamma, v, x)$ where
\begin{itemize}
\item $\gamma$ is the supercritical state of $I$
\item $v \in [k]^{2^h - 1 - h}$ is a string that specifies all of the correct node values except the first $h$ values learned on $C_1(I, \gamma)$. In particular, $v = u_1 u_2$, where $u_1$ specifies the correct values of the unlearned nodes queried on $C_1(I, \gamma)$ \emph{in order of their first occurrence}, and $u_2$ specifies the correct values of the remaining nodes.
\item $x \in [k]^{(k^2 - 1)(2^{h-1} - 1)}$ is a string that specifies the values of the non-thrifty queries for $I$.
\end{itemize}

Here is the crucial lemma needed to prove the lower bound.

\begin{lemma}
\label{lem:td_tag}
The tagging function $U$ is one-to-one.
\end{lemma}
\begin{proof}
Let $I$ and $J$ be inputs such that $U(I) = U(J) = (\gamma, v, x)$. 
We first claim that for every state $\delta$ on $C_1(I, \gamma)$, $J$ follows $C_1(I, \gamma)$ up to $\delta$, and for every node $i$ that is either queried or learned between $\gamma$ and $\delta$, $v^I_i = v^J_i$. This is vacuously true for $\delta = \gamma$. Now pick some $\delta$ which isn't the output state, and suppose the claim holds. Consider the two possibilities for the node $i$ queried at $\delta$:
\begin{itemize}
\item If $i$ has been queried or learned between $\gamma$ and $\delta$, then $v_i^J = v_i^I$, and because this is a thrifty query, $J$ must follow the same edge as $I$ out of $\delta$.
\item Otherwise, the edge $J$ follows out of $\delta$ is specified by $v$ from the tag. The specific position is identical to that of the tag for $I$ because the computation paths are identical up to this point. Since $I$ and $J$ have the same tag, they again follow the same edge out, and so $v_i^I = v_i^J$.
\end{itemize}
Suppose node $j$ is learned at $\delta$ for $I$; since the conditions necessary for learning values depends only on the segment of $C_1(I, \gamma)$ before, which is followed by both $I$ and $J$, the two inputs both learn the value of $j$. Since this value is determined only by $\delta$, $v^I_j = v^J_j$.

It follows by induction that $C_1(I, \gamma) = C_1(J, \gamma)$, and moreover that $I$ and $J$ agree on the correct values of all nodes which are either queried or learned after $\gamma$. $I$ and $J$ have the exact same nodes which are neither queried nor learned after $\gamma$, and then because their tags are identical, they agree on the correct values of these nodes as well. Finally, $I$ and $J$ agree on all of their non-thrifty queries, which are completely specified by $x$ in the tag.
\end{proof}

\begin{theorem}[\cite{cook12}]
Every deterministic thrifty branching program solving $TEP^h_2(k)$ has at least $k^h$ states.
\end{theorem}
\begin{proof}
Together, $v$ and $x$ specify all by $h$ of the input values. Since $U$ is one-to-one, there must be $k^h$ different choices for $\gamma$, the supercritical state. Another way to say this is that at most $1/k^h$ inputs can have the same supercritical state.
\end{proof}

\section{Non-Deterministic Read-Once Thrifty}
In this section, we will present our first new result, a lower bound for non-deterministic, read-once thrifty branching programs. We first use the notion of \emph{syntactic} read-once, forcing every complete path to query each input value at most once, regardless of whether the path can be followed by some input. After we present the pebbling argument for this restriction, we will replace this read-once restriction with a slightly more general one which suffices to give the same bound. Finally, we show how to extend the argument to \emph{semantic} read-once branching programs, at the expense of a constant factor.
We remark that our lower bounds are tight when $h$ is even when $h$ is even, because in this case the black-white \emph{whole} pebbling number of $T^h_2$ coincides with the \emph{fractional} one (in Theorem \ref{thm:frac_num}).

\subsection{Syntactic Read-Once}

It is easy to show that every complete computation path on a non-deterministic thrifty branching program must query each node of $T^h_2$, in a manner similar to Proposition \ref{prop:td_crit}. This fact combined with the read-once restriction means that every complete computation path corresponds to a \emph{permutation} of the nodes. We now follow the proof strategy outlined in the previous section, using this fact to define a black-white whole pebbling sequence.

One key idea used in this proof is a composability property of syntactic read-once branching programs. Essentially, this says that given two complete computation paths through a state $\gamma$, it is possible to ``switch" between them at $\gamma$. Note that this property does \emph{not} hold for semantic read-once branching programs. 
\begin{proposition}[Composability]
\label{prop:comp}
Let $C(I)$ and $C(J)$ be two complete computation paths on a syntactic read-once branching program that both contain some state $\gamma$. Then there exists an input $K$ which has a complete computation path $C(K)$ that follows $C_0(I, \gamma)$ up to $\gamma$, and then $C_1(J, \gamma)$ after $\gamma$.
\end{proposition}
\begin{proof}
The read-once restriction implies that $C_0(I, \gamma)$ and $C_1(J, \gamma)$ never make the same query. Therefore we can choose $K$ to agree with $I$ on the queries before $\gamma$, and with $J$ on the queries at and after $\gamma$.
\end{proof}

\subsubsection{Pebbling Sequence}
We will associate a pebbling to each complete computation path $C(I)$, with the intuition that black pebbles represent values \emph{known} to be correct because they've already been queried, and white pebbles represent guessed values that \emph{should} be correct due to thriftiness, but have yet to be queried (i.e., verified).

We use the following rules to associate pebbling configurations to the states on $C(I)$ in order. Unlike the previous section, we may perform more than one move at a state, and hence multiple configurations may be associated with the same state.
For a state $\gamma$ on $C(I)$ querying node $i$, these steps are performed in order:

\begin{enumerate}
\item[(1)] Place white pebbles on any children of $i$ that are currently unpebbled. The configurations produced are associated with state $\gamma$. 
\item[(2a)] If $i$ is pebbled, it must be white-pebbled. Remove the pebble from $i$, and remove all black pebbles on the children of $i$. The configurations produced are associated with $\gamma$.
\item[(2b)] If $i$ is not pebbled, put a black pebble on $i$, and simultaneously remove all black pebbles on the children of $i$. The configuration produced is associated with the state \emph{following} $\gamma$.
\end{enumerate}
When an internal node is queried, the thriftiness condition ensures that the values of the children of $i$ must be guessed if they haven't yet been queried; and if they have been queried, their values can be forgotten because $i$ will only be queried once.
The result of the query is either remembered by the following state, or is a verification of a previous guess, after which the guess can be forgotten.
Given that the starting configuration is empty, it is easy to check that these rules generate a valid pebbling sequence, except that the last configuration has a black pebble on the root. Performing a final move of removing the black pebble results in a valid pebbling sequence.

The following proposition connects the pebbles with the order of queries on $C(I)$.

\begin{proposition} 
\label{prop:rot_order}
Let $C(I)$ be a complete computation path on a non-deterministic, syntactic read-once, thrifty branching program solving $TEP^h_2(k)$. Let $\gamma$ be any state in $C(I)$, and let $\mathcal C$ be any (not necessarily the last) pebbling configuration associated with $\gamma$ according to the above rules. For each node $i$:
\begin{itemize}
\item If $i$ has a black pebble at $\mathcal C$, then it has been queried on $C_0(I, \gamma)$, and its parent (if it is not the root) is queried on $C_1(I, \gamma)$.
\item If $i$ has a white pebble at $\mathcal C$, then it is queried on $C_1(I, \gamma)$, and its parent has been queried on $C_0(I, \gamma)$ or is queried at $\gamma$ itself. This precludes the root from being white-pebbled.
\end{itemize}
\end{proposition}
\begin{proof}
In the above rules, a black pebble is placed on a node on the state immediately after the one querying that node. Also, a black pebble is always removed at the state that queries its parent, so it must be that if $i$ is black pebbled at $\gamma$, its parent has not yet been queried. Because $C(I)$ is a permutation of the nodes, its parent must be queried at some point on $C(I)$, and hence this must be after $\gamma$.

Similarly, a white pebble is placed on a node only when its parent is queried, and is only removed when it is queried.
\end{proof}

\subsubsection{The Lower Bound}
The following critical lemma establishes the significance of the pebbles: states must ``remember" the correct values of nodes that are pebbled at their configurations.

\begin{lemma}
\label{lem:rot_pebbles}
Let $C(I)$ be a complete computation path on a non-deterministic, syntactic read-once, thrifty branching program solving $TEP^h_2(k)$.
Suppose node $i$ is pebbled on state $\gamma$ of $C(I)$; then every complete computation path $C(J)$ through $\gamma$ must satisfy $v_i^J = v_i^I$.
\end{lemma}
\begin{proof}
Let $C(J)$ be any complete computation path through $\gamma$. 
By Proposition \ref{prop:comp}, we can choose a complete computation path $C(K)$ such that $C_0(K, \gamma) = C_0(I, \gamma)$ and $C_1(K, \gamma) = C_1(J, \gamma)$.

First suppose node $i$ has a black pebble at $\gamma$ (with respect to $I$). Then by Proposition \ref{prop:rot_order}, $C_0(I,\gamma)$ queries $i$ and $C_1(I, \gamma)$ queries its parent. 
Then $v_i^K = v_i^I$ because $C_0(K,\gamma) = C_0(I,\gamma)$. Since $C_1(K,\gamma) = C_1(J,\gamma)$ queries the parent of $i$, $v^K_i = v^J_i$ due to thriftiness.

Now suppose $i$ has a white pebble at $\gamma$ (with respect to $I$). If $\gamma$ queries the parent of $i$, then $v_i^I = v_i^J$ by thriftiness. Otherwise, by Proposition \ref{prop:rot_order} again, the parent of $i$ is queried on $C_0(I, \gamma)$, and hence $v_i^K = v_i^I$. But $i$ itself is queried on $C_1(K, \gamma) = C_1(J,\gamma)$, and so $v_i^K = v_i^J$.
\end{proof}

We will use this lemma in combination with the known pebble number for $T^h_2$ to derive a lower bound for the size of these branching programs.

\begin{proof}[Proof of Theorem \ref{thm:ro_ndet}, syntactic read-once case]
We have established that any complete computation path $C(I)$ may be associated with a valid black-white whole pebbling. Thus we may define the \emph{supercritical state} of $C(I)$ to be the first state that has an associated pebble configuration with at least $\ceil{\frac{h}{2}} + 1$ pebbles, using the bound of Theorem \ref{thm:frac_num}. By Lemma \ref{lem:rot_pebbles}, these pebbles determine at least $\ceil{\frac{h}{2}} + 1$ of the correct node values. Thus if we define a map each input to its supercritical state, at most $1/k^{\ceil{\frac{h}{2}} + 1}$ of the inputs can be mapped to the same state, and the theorem follows.
\end{proof}

We remark that this proof is implicitly uses the same ``tag'' argument as the previous section; however, since the tag here only needs two components (the supercritical state and a string specifying all of the other input values), we omitted the notation.

\subsection{Null-Path-Free}
Recall that a \emph{null-path} in a non-deterministic branching program is a complete path that is inconsistent, i.e., that has two states that query the same variable but takes edges with different labels out of each state. Such paths are ``useless" in the sense that no input can follow them to an output state; however, it is known that their presence can result in an exponential decrease in the size of branching programs for certain problems \cite{juk13}. A branching program is \emph{null-path-free} if it contains no null-paths.
Every syntactic read-once branching program is null-path-free.

We now generalize the previous argument to non-deterministic thrifty null-path-free branching programs. Crucially, composability still applies. Moreover, even though states may now be queried more than once, considering only the \emph{first} time each node is queried yields a permutation; we will call these states the \emph{critical states} for nodes. We can then apply the same pebbling moves as before to the critical states, associating the configurations with these states (and ignoring all other states on the computation path). The rules still generate a valid pebbling sequence, and the following analogue of Proposition \ref{prop:rot_order} holds by a very similar argument.

\begin{proposition}
\label{prop:npft_order}
Let $C(I)$ be a complete computation path on a non-deterministic, null-path-free, thrifty branching program solving $TEP^h_2(k)$. Let $\gamma$ be a critical state in $C(I)$, and let $\mathcal C$ be any (not necessarily the last) pebbling configuration associated with $\gamma$ according to the rules from the previous Section. For each node $i$:
\begin{itemize}
\item[(1)] If $i$ has a black pebble at $\mathcal C$, then it has been queried on $C_0(I, \gamma)$, and its parent (if it is not the root) is queried \emph{for the first time} on $C_1(I, \gamma)$.
\item[(2)] If $i$ has a white pebble at $\mathcal C$, then it is queried \emph{for the first time} on $C_1(I, \gamma)$, and its parent has been queried on $C_0(I, \gamma)$ or is queried at $\gamma$ itself. This precludes the root from being white-pebbled.
\end{itemize}
\end{proposition}
We now prove the key technical lemma, whose statement remains unchanged from the previous subsection, while the proof contains but a few subtle differences. The lower bound then follows directly, as before.

\begin{lemma}
\label{lem:npft_pebbles}
Let $C(I)$ be a complete computation path on a non-deterministic, null-path-free, thrifty branching program solving $TEP^h_2(k)$.
Suppose node $i$ is pebbled on a critical state $\gamma$ of $C(I)$; then every complete computation path $C(J)$ through $\gamma$ must satisfy $v_i^J = v_i^I$.
\end{lemma}
\begin{proof}
Let $C(J)$ be any complete computation path through $\gamma$. 
Since composability still holds, we can choose a complete computation path $C(K)$ agreeing with $C(I)$ before $\gamma$ and $C(J)$ after $\gamma$.

First suppose node $i$ has a black pebble at $\gamma$ (with respect to $I$). Then by Proposition \ref{prop:npft_order}, $C_0(I,\gamma)$ queries $i$, hence $v_i^K = v_i^I$. 
If $C_1(K,\gamma) = C_1(J,\gamma)$ queries the parent of $i$, then by thriftiness $v_i^K = v_i^J$. Otherwise, $C_0(J,\gamma)$ queries the parent of $i$, so consider instead the computation path $C(K')$ which uses $C_0(J,\gamma)$ and $C_1(I,\gamma)$. Then $C_0(K',\gamma)$ queries the parent of $i$, and by Proposition \ref{prop:npft_order} so does $C_1(K',\gamma) = C_1(I,\gamma)$, and by thriftiness these are the same query, and so $v_i^J = v_i^I$.

Now suppose $i$ has a white pebble at $\gamma$ (with respect to $I$). If $\gamma$ queries the parent of $i$, then $v_i^I = v_i^J$ by thriftiness. Otherwise, by Proposition \ref{prop:npft_order} again, the parent of $i$ is queried on $C_0(I, \gamma)$, and hence $v_i^K = v_i^I$. If $i$ is queried on $C_1(J, \gamma)$ then $v_i^K = v_i^J$ and we are done. Otherwise, $i$ is queried on $C_0(J,\gamma)$ and (by Proposition \ref{prop:npft_order}) $C_1(I,\gamma)$. By composability, there is a computation path $C(K')$ agreeing with $C_0(J,\gamma)$ and $C_1(I,\gamma)$; by thriftiness it makes the same query to node $i$ on both segments, and by the null-path-free property the result is the same, hence $v_i^I = v_i^J$.
\end{proof}

Thus Theorem \ref{thm:ro_ndet} holds for the null-path-free case, by analogy to the syntactic read-once case.

\subsection{Semantic Read-Once}
This final subsection deals with the weaker \emph{semantic} read-once restriction, where only complete \emph{computation} paths (i.e., paths can can be followed by some input) must be read-once. We will show that the same black-white whole pebbling argument holds in this setting, up to a constant factor in the lower bound.

First, the observation that a complete computation path corresponds to a permutation of the nodes still holds for semantic read-once branching programs. Therefore we can apply the same pebbling rules as before, and associate to each complete computation path $C(I)$ a \emph{supercritical state} $\gamma$, the first state on $C(I)$  whose associated pebbling configuration has $\ceil{\frac{h}{2}} + 1$ pebbles. We now proceed slightly differently, returning to the tag argument of Section 3. Define the tagging function $U(I) = (u, \gamma, x)$, where:
\begin{itemize}
\item $u$ is a number encoding which permutation corresponds to $C(I)$.
\item $\gamma$ is the supercritical state of $C(I)$.
\item $x \in [k]^{2^h - 1 + (k^2 - 1)(2^{h-1} - 1) - \ceil{\frac{h}{2}} - 1}$ specifies all of the input values except the correct values of the nodes pebbled at $\gamma$.
\end{itemize}

Here is the main technical lemma that allows us to prove the lower bound. 

\begin{lemma}
The tagging function $U$ is one-to-one.
\end{lemma}
\begin{proof}
Let $I$ and $J$ be two inputs with corresponding complete computation paths $C(I)$ and $C(J)$, and suppose $U(I) = U(J) = (u, \gamma, x)$. First, because $C(I)$ and $C(J)$ correspond to the same permutation, and both pass through $\gamma$, the segments $C_0(I,\gamma)$ and $C_0(J,\gamma)$ query the same set of nodes, as do $C_1(I,\gamma)$ and $C_1(J,\gamma)$. Then the sets of queries made by $C_0(I,\gamma)$ and $C_1(J,\gamma)$ are disjoint, and therefore composability holds: there is an input $K$ with a complete computation path $C(K)$ that follows $C_0(I,\gamma)$ and then $C_1(J,\gamma)$.
The same argument as Lemma \ref{lem:rot_pebbles} shows that $I$ and $J$ agree on the values of all $\ceil{\frac{h}{2}} + 1$ pebbled nodes at $\gamma$, and this together with the $x$ part of the tag implies that $I = J$.
\end{proof}
Since the number of possible permutations depends on $h$ but not $k$, the asymptotic lower bound follows in Theorem \ref{thm:ro_ndet} for the semantic read-once restriction.

\section{Read-Once Deterministic}

In this section, we present our most complex result, a lower bound for deterministic, read-once branching programs solving the Tree Evaluation Problem. Note that for deterministic branching programs, there is no difference between the syntactic and semantic notions of read-once. While the overall strategy used to derive the lower bound is the same as before, each component is more complicated than the two previous sections.

Our first goal is to take a computation path $C(I)$ and assign a pebbling sequence to its states. 
Unlike thrifty branching programs, where the behaviour of inputs depends only on their correct node values, now states must be able to encode information about non-thrifty queries. 
The high level idea is the same as before: we define an algorithm that takes a complete computation path $C(I)$, processes the states one by one, and outputs a sequence of pebbling configurations. 
Each state is processed in two phases. 
The result of the state's query is used to update several pieces of ``auxiliary data" used by the algorithm, and then this data is used to make a sequence of pebbling moves.
This algorithm is deterministic; at any point, the values of the auxiliary data and pebbling configuration depend only on the states and edges that have been processed so far.
Before we describe the algorithm itself, we define the variations in the pebbling and the auxiliary data used by the algorithm.

\subsection{Definitions}
For each node $i$ and $a \in [k]$, the \emph{logical variable} $[i, a]$ represents the statement ``$v_i = a$".
A computation path can be interpreted as a derivation of a singleton formula $[1, a]$, where $a$ corresponds to the output state that is reached.
There are two types of pebbles used to represent the information encoded by a state:
\begin{itemize}
\item {\bf Grey pebbles} represent an implication among the variables as a result of queries. 
These have labels of the form 
$[2i, a] \wedge [2i + 1, b] \imp [i, c]$, where $i$ is the pebbled node.
\item {\bf Black pebbles} represent node values that must be correct. These have labels of the form $[i, a]$, where $i$ is the pebbled node. 
\end{itemize}
Intuitively, grey pebbles represent information about queries which could be thrifty or non-thrifty, and black pebbles remain encodings of node values which are known to be correct.

The auxiliary data is the set of objects defined below.

\begin{definition}[Range] 
Let $\gamma$ be a state on a computation path $C(I)$. The set $Range^I_\gamma(i) \subseteq [k]$ stores the possible ``correct values" of node $i$ at $\gamma$.
That is, $a \in Range^I_\gamma(i)$ if and only if there exists an input $I'$ that agrees with $I$ before $\gamma$ and has $v_i^{I'} = a$. When the context is clear, we will drop the $I$ and write $Range_\gamma(i)$.
\end{definition}

\begin{definition}[Gap, Completely Queried]
Let $\gamma$ be a state on a computation path $C(I)$. A \emph{gap} for an internal node $i$ at $\gamma$ on $C(I)$ is a function value $f_i(x,y)$ where $x \in Range^I_\gamma(2i), y \in Range^I_\gamma(2i + 1)$, and the function value was not queried by $C(I)$ before $\gamma$. $I$ has a gap for a leaf $i$ if $v_i$ has not yet been queried.

If $i$ has no gaps at $\gamma$, it is \emph{completely queried} at $\gamma$ with respect to $I$.
\end{definition}

The next two definitions capture what is \emph{yet to be learned} at $\gamma$.

\begin{definition}[Equivalence]
Let $\gamma$ be a state on a computation path $C(I)$. For $a_1, a_2 \in Range^I_\gamma(i)$, we define equivalence between these corresponding variables, denoted $[i, a_1] \approx^I_\gamma [i,a_2]$,
according to the following recursive definition:
\begin{itemize}
\item $[1, a_1] \approx^I_\gamma [1, a_2] \iff a_1 = a_2$.
\item Let $i$ be a non-root node with sibling $i'$ and parent $j$. 
Then $[i, a_1] \approx^I_\gamma [i, a_2]$ if and only if for all $b \in Range^I_\gamma(i')$, $f_j(a_1, b)$ and $f_j(a_2, b)$ have been queried on $C_0(I,\gamma)$, and $[j, f_j(a_1, b)] \approx^I_\gamma [j, f_j(a_2, b)]$.
\end{itemize}
\end{definition}
One can check that this is an actual equivalence relation for each state and node.
Intuitively, two variables are equivalent if changing the node's correct value from one to the other doesn't change the correct root value.
On the other hand, if a node has two or more equivalence classes at $\gamma$, $C_1(I,\gamma)$ should make more queries to the node or its descendants to determine the correct equivalence class.
This motivates the next definition.

\begin{definition}[Node Activity]
Let $\gamma$ be a state on a computation path $C(I)$. 
A node $i$ is \emph{active} for $I$ at $\gamma$ if its parent is active and there exist $a_1, a_2 \in Range^I_\gamma(i)$ such that $[i, a_1] \not \approx^I_\gamma [i, a_2]$.
Otherwise, the node is \emph{inactive} at $\gamma$.
\end{definition}

All nodes begin active at the initial state, because all variables are inequivalent. A computation path may be interpreted as a sequence of queries made until the root becomes inactive. The following intuition may be helpful later on. 
There are two ways for a node to become inactive: either its correct value is found (i.e., $|Range_\gamma(i)| = 1$), or it or one of its ancestors has been completely queried and that node's equivalence classes merged into one.
The former can be done through thrifty queries (``efficiently"), while the latter requires many queries.

The final definition encompasses all of the previous ones, formalizing the notion of auxiliary data.

\begin{definition}[Memory]
The \emph{memory} of the pebbling algorithm at $\gamma$ relative to path $C(I)$ is the set of queries that have been made on $C_0(I, \gamma)$ and their results (put another way, the states and the edges traversed).
From this the pebbling algorithm can calculate the Ranges, equivalence classes, and node activity.
Thus when we refer to the algorithm's memory at a point in time, we implicitly include these three properties for every node.
\end{definition}

The pebbling configuration output immediately before processing $\gamma$ is the \emph{configuration associated with} $\gamma$.
As we will discuss in the next subsection, we distinguish between the memory and the configuration of a state because of their different functions in the pebbling algorithm. 
Roughly speaking, the memory is the internal storage used by the algorithm, while the pebble configurations are the actual output.

\begin{algorithm}
\caption{Pebbling Algorithm}
\begin{algorithmic}[1]
\Require $C(I)$, a computation path in the branching program.
\Ensure A sequence of pebble configurations corresponding to $C(I)$.
\hrule
\vspace{5pt}

\ForAll {states $\gamma$ on $C(I)$ in path order}
	\ForAll {nodes $i$, in bottom-up order}
		\If {$i$ is a leaf and has been queried on $C_0(I, \gamma)$ with value $a$}
			\State $Range_\gamma(i) \gets \{a\}$.
		\ElsIf {$i$ is an internal node and completely queried at $\gamma$}
			\State $Range_\gamma(i) \gets \{f_i(x,y) \mid (x,y) \in Range_\gamma(2i) \times Range_\gamma(2i + 1)\}$.
		\Else
			\State $Range_\gamma(i) \gets [k]$
		\EndIf
	\EndFor
		
	\State Update equivalence classes using the recursive definition (this is top-down).
	\State Update activity of each node using the definition.
	
\vspace{10pt}

	\If {the state preceding $\gamma$ queries leaf $i$ with value $a$}
		\State Place a black pebble $[i, a]$ (on $i$).
	\ElsIf {the state preceding $\gamma$ queries function $f_i(a, b)$ with result value $c$}
		\State Place a grey pebble $[2i, a] \wedge [2i+1, b] \imp [i, c]$ (on $i$).
	\EndIf
	
\vspace{5pt}

	\ForAll {nodes $i$, in bottom-up order}
		\If {$i$ is inactive}
			\State Remove all grey pebbles from $i$.
			\If {$Range_\gamma(i) = \{a\}$ for some $a$ and $i$ hasn't been previously black pebbled}
				\State Place a black pebble $[i, a]$ (on $i$).
			\EndIf
		\EndIf
		\If {$i$ is inactive or completely queried}
				\State Remove any black pebbles on the children of $i$.
		\EndIf
		
		\ForAll {$a \notin Range_\gamma(i)$}
			\State Remove all grey pebbles with $[i, a]$ in antecedent.
		\EndFor
	\EndFor
	\State Associate the latest pebbling configuration produced with $\gamma$.
\EndFor

\vspace{5pt}

\State \Return the sequence of pebbling configurations produced.
\end{algorithmic}
\end{algorithm}

\subsection{Description of the Algorithm}
Lines 2-12 comprise the ``computation" done by the algorithm.
First, path segment $C_0(I,\gamma)$ is used to update the $Range_\gamma(i)$ (Lines 2-10), equivalence classes (Line 11), and node activity (Line 12).
Note that this is all done independently of the pebbling configuration, which is only updated in the second phase.

The remaining steps (Lines 13-33) use the updated memory to produce new pebble configurations.
First, a new pebble gets placed as a result of the query (Lines 13-17).
Lines 20 and 25-27 remove all pebbles from inactive nodes, except for when a black pebble is still necessary to determine if a future query to the parent could be the thrifty query.
Line 22 essentially replaces a grey pebble with a black pebble, as this is the only scenario when additional black pebbles can be placed.
Note that we allow the algorithm to also remember when nodes have been black pebbled, so that each node gets black-pebbled at most once during the algorithm.
Finally, Lines 28-30 remove the grey pebbles which correspond to queries that cannot possibly be thrifty, and hence need not be remembered.

\subsection{Properties of the Pebbling}
In the following proofs, we will often use the fact that if two inputs $I$ and $I'$ agree before $\gamma$, then they must have the same memory at $\gamma$. We also remind the reader that Proposition \ref{prop:comp} still holds; the following proofs will frequently compose computation paths.

We first prove the basic correctness of the pebbling algorithm in relation to the definition of $Range$ and the intuitive significance of black and grey pebbles.

\begin{proposition}
\label{prop:range}
Let $\gamma$ be a state on a computation path $C(I)$ on a deterministic, read-once branching program solving $TEP^h_2(k)$. 
\begin{enumerate}
\item[(i)] For all inputs $I'$ agreeing with $I$ before $\gamma$, and all nodes $i$, $v_i^{I'} \in Range_\gamma^I(i)$.
\item[(ii)] Let $C$ be any collection of nodes such that none is an ancestor of any other. For each $i \in C$ let $a_i \in Range^I_\gamma(i)$. 
Then there exists an input $I'$ agreeing with $I$ before $\gamma$ such that $v_i^{I'} = a_i$ for all $i \in C$.
\item[(iii)] If the pebbling configuration at $\gamma$ has black pebble $[i, a]$, then $v^I_i = a$. 
\item[(iv)] If the pebbling configuration has grey pebble $[2i, a] \wedge [2i+1, b] \imp [i, c]$ then there exists an input $I'$ that agrees with $I$ before $\gamma$, and $v_{2i}^{I'} = a$, $v_{2i+1}^{I'} = b$, and $v_i^{I'} = c$.
\end{enumerate}
\end{proposition}
\begin{proof}[Proof of (i)]
We only need to consider the case where $Range^I_\gamma(i) \neq [k]$.
If $i$ is a leaf, then $Range^I_\gamma(i) = \{a\}$ where $a = v_i^I$ was queried by $C(I)$ before $\gamma$, and hence $v_i^{I'} = a \in Range^I_\gamma(i)$.
Suppose $i$ is an internal node. By induction, $v_{2i}^{I'} \in Range^I_\gamma(2i)$ and $v_{2i+1}^{I'} \in Range^I_\gamma(2i+1)$. 
If $Range^I_\gamma(i) \neq [k]$ then $i$ is completely queried at $\gamma$, and in particular $f_i(v_{2i}^{I'}, v_{2i+1}^{I'}) = v_i^{I'} \in Range^I_\gamma(i)$. 
\end{proof}

\begin{proof}[Proof of (ii)]
Let $i \in C$.
If $i$ is a leaf and was queried on $C_0(I,\gamma)$, any $I'$ agreeing with $I$ before $\gamma$ has $v^{I'}_i = v^I_i$, and $Range^I_\gamma(i) = \{v^I_i\}$.
If $i$ is a leaf and wasn't queried before $\gamma$, then we can choose $I'$ to have $v_i^{I'} = a_i$ and still agree with $I$ before $\gamma$.
Suppose $i$ is an internal node.
If there a gap $f_i(x,y)$ at $\gamma$ for $I$, we can choose $f^{I'}_i (x,y) = a_i$.
Otherwise, $i$ is completely queried at $\gamma$ and there must exist $x \in Range^I_\gamma(2i)$ and $y \in Range^I_\gamma(2i+1)$ such that $f_i (x,y) = a_i$. 
Setting $v_{2i}^{I'} = x$ and $v_{2i+1}^{I'} = y$ by induction results in $v_i^{I'} = a_i$.

Now we observe that such an $I'$ can be found for $a_i$ by choosing certain function and/or leaf values in the subtree rooted at $i$, none of which were queried on $C_0(I,\gamma)$.
This can be done independently for each $i \in C$ because they have disjoint subtrees.
\end{proof}

\begin{proof}[Proof of (iii)]
This follows immediately from (i), since $i$ has a black pebble only if $|Range(i)| = 1$. 
\end{proof}
\begin{proof}[Proof of (iv)]
This follows immediately from applying (ii) to $[2i, a]$ and $[2i + 1, b]$.
Note that $a \in Range^I_\gamma(2i)$ and $b \in Range^I_\gamma(2i+1)$, as otherwise this grey pebble would have been removed by Line 29.
\end{proof}

The following two propositions illustrate the significance of the active nodes.
\begin{proposition}
\label{prop:equiv_1}
Let $\gamma$ be a state on a computation path $C(I)$ on a deterministic, read-once branching program solving $TEP^h_2(k)$. 
Let $i_0$ be an active node at $\gamma$, with $a_1, a_2 \in Range_\gamma(i_0)$ such that $[i_0, a_1] \not \approx^I_\gamma [i_0, a_2]$.
Let $i_1, \dots, i_m = 1$ be the ancestors of $i_0$, and let $i'_l$ denote the sibling of $i_l$ for $l = 0, \dots, m-1$.
Then there exist $b_0 \in Range_\gamma(i'_0), \dots, b_{m-1} \in Range_\gamma(i'_{m-1})$ and an input $I_1$ with the following properties:
\begin{itemize}
\item $I_1$ agrees with $I$ before $\gamma$.
\item $a_1, b_0, \dots, b_{m-1}$ are the correct values of the corresponding nodes for $I_1$.
\item Let $I_2$ to be identical to $I_1$ except possibly on the subtree rooted at $i_0$, and having $v^{I_2}_{i_0} = a_2$. Then $v_1^{I_2} \neq v_1^{I_1}.$
(Note that this is a claim about the input and not the state $\gamma$; i.e., we are not claiming here that $I_2$ reaches $\gamma$.)
\end{itemize}
\end{proposition}

\begin{figure}[h]
\begin{center}
\includegraphics[height=100pt]{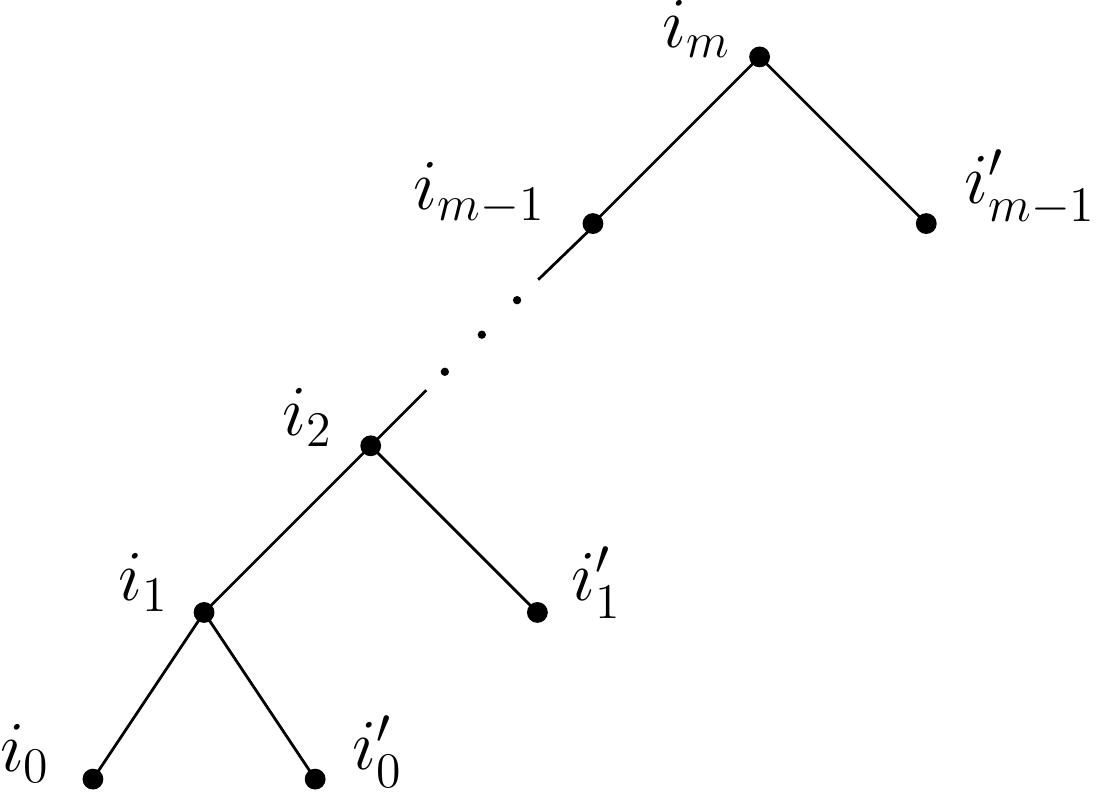}
\caption{Illustration of the notation in Proposition \ref{prop:equiv_1}.}
\end{center}
\end{figure}

\begin{proof}
If $i_0$ is the root, simply take $I_0$ to have $a_1$ be the correct root value as in Proposition \ref{prop:range}(ii).
Suppose $i_0$ is a non-root node. 
If $f_{i_1}(a_1, b)$ and $f_{i_1}(a_2, b)$ have been queried for all $b \in Range_\gamma(i'_0)$, then because $[i_0, a_1] \not \approx [i_0, a_2]$, there exists $b_0 \in Range_\gamma(i'_0)$ such that $[i_1, f_{i_1}(a_1, b_0)] \not \approx [i_1, f_{i_1}(a_2, b_0)]$.
By Proposition \ref{prop:range}(ii), there exists an input $I_1$ that agrees with $I$ before $\gamma$ and has $a_1$ and $b_0$ the correct values of $i_0$ and $i'_0$, respectively.
We can then use induction on $i_1$.
Now suppose there exists $b_0 \in Range_\gamma(i'_0)$ such that $f_{i_1}(a_1, b_0)$ has been queried with result value $c_1$ and $f_{i_1}(a_2, b_0)$ hasn't been queried.
Because $i_1$ is active, it has more than one equivalence class.
Let $c_2 \in Range_\gamma(i_1)$ with $[i_1, c_1] \not \approx [i_1, c_2]$.
Therefore we can choose $I_1$ to have $f_{i_1}(a_2, b_0) = c_2$, and use induction.
A similar argument holds for when neither $f_{i_1}(a_1, b_0)$ nor $f_{i_2}(a_2, b_0)$ have been queried.
\end{proof}

\begin{proposition}
\label{prop:equiv_2}
Let $\gamma$ be a state on a computation path $C(I)$ on a deterministic, read-once branching program solving $TEP^h_2(k)$.
Let $i$ be an active node at $\gamma$. 
Then there exist inputs $I_1$ and $I_2$ that agree with $I$ before $\gamma$, differ only on the subtree rooted at $i$, but have different correct root values.
Moreover, if $I$ has a gap $f_i(x,y)$, then $I_1$ and $I_2$ can be made to differ \emph{only} at the value of $f_i(x,y)$.
\end{proposition}
\begin{proof}
Choose $a_1,a_2 \in Range_\gamma(i)$ and $I_1$ as in Proposition \ref{prop:equiv_1} (identifying $i$ with $i_0$).
Note that we can apply Proposition \ref{prop:range}(ii) to make $a_2$ the correct value of node $i_0$ instead of $a_1$ by changing only the queries of the subtree rooted at $i_0$, and this yields the other input $I_2$.
If $I$ had a gap at $f_i(x,y)$, simply take $I_2$ to be $I_1$ except $f_i(x,y) = a_2$ instead of $a_1$, making $x$ and $y$ the correct values of the children of $i$ for $I_1$ and $I_2$ by Proposition \ref{prop:range}(ii).
\end{proof}

\begin{proposition}
\label{prop:parent_gap}
Let $\gamma$ be a state on a computation path $C(I)$ on a deterministic, read-once branching program solving $TEP^h_2(k)$.
Let $i$ be any non-root node.
If $|Range^I_\gamma(i)| > 1$ and the parent of $i$ is active and has a gap at $\gamma$, then $i$ is active.
\end{proposition}
\begin{proof}
This follows from the definition of equivalence, since if the parent $j$ of $i$ has a gap $f_j(x,y)$, then (assuming $i$ is the left child) $[i,x]$ is its own equivalence class.
\end{proof}

\subsection{Meaningful Pebbles}

As previously observed, fixing the entire computation path up to a state $\gamma$ determines the memory encoded at $\gamma$.
The goal of the remainder of this subsection is to show that the state $\gamma$ \emph{alone} is enough to almost completely specify this information. First we need a few more easy properties, taking advantage of the read-once restriction.

\begin{proposition} \label{prop:same_gaps}
Let $\gamma$ be a state on a computation path $C(I)$ on a deterministic, read-once branching program solving $TEP^h_2(k)$.
Let $i$ be an active node with a gap at $\gamma$. 
Then for every input $J$ that reaches $\gamma$, $C_0(J,\gamma)$ does not query this gap.
\end{proposition}
\begin{proof}
By Proposition \ref{prop:equiv_2} there are two inputs $I_1$ and $I_2$ that agree with $I$ before $\gamma$, have different correct root values, and differ only at the value of the gap query.
Then $I_1$ and $I_2$ must query this value after $\gamma$ to be able to reach different output states.
The claim then follows by the read-once property.
\end{proof}

\begin{proposition}
\label{prop:same_active}
Let $\gamma$ be a state on two computation paths $C(I)$ and $C(J)$ on a deterministic, read-once branching program solving $TEP^h_2(k)$.
Each node $i$ is active at $\gamma$ for $I$ if and only if it is active for $J$.
\end{proposition}
\begin{proof}
Suppose there is a node $i$ which is active for $I$ and not $J$. 
By Proposition \ref{prop:equiv_2}, there must exist $I_1$ and $I_2$ which differ only at the subtree rooted at $i$, agree with $I$ before $\gamma$, and have different root values.
Using composability, we can define two inputs $K_1$ and $K_2$ that agree with $J$ before $\gamma$, and with $I_1$ and $I_2$ after $\gamma$, respectively.
Then $C(K_1)$ and $C(K_2)$ have different output states even though $v_1^{K_1} = v_1^{K_2}$, as they differ only on an inactive subtree.
\end{proof}

\begin{proposition}
\label{prop:same_black}
Let $\gamma$ be a state on two computation paths $C(I)$ and $C(J)$ on a deterministic, read-once branching program solving $TEP^h_2(k)$.
If $I$ has a black pebble $[i, a]$ at $\gamma$, then so does $J$.
\end{proposition}
\begin{proof}
Node $i$ is inactive for $I$, and hence by Proposition \ref{prop:same_active}, it must be inactive for $J$. 
If $i$ is the root, then $J$ must also have a black pebble on the root, because the root can only be inactive at $\gamma$ if $|Range_\gamma(1)| = 1$, and black pebbles are never removed from the root.
Let $[1, a']$ be the label of the black pebble at $\gamma$ for $J$.
Let $K$ follow $C_0(I, \gamma)$ and then $C_1(J, \gamma)$.
Then $v_1^K = a$ because it agrees with $I$ before $\gamma$, and $C(K)$ reaches output state $a'$, hence $a = a'$.

Now suppose $i$ is a non-root node, and let $i'$ and $j$ be the sibling and parent of $i$, respectively.
The fact that $i$ is black pebbled for $I$ means that $j$ is active and $I$ has some gap $f_j(a, b)$ at $\gamma$ (otherwise the black pebble would have been removed).
Then by Proposition \ref{prop:equiv_2} there exist two inputs $I_1, I_2$ that agree with $I$ before $\gamma$, differ \emph{only} at the value of $f_j(a, b)$, and that have different correct root values.
Set $K_1$ to agree with $J$ before $\gamma$ and $I_1$ after $\gamma$, and similarly define $K_2$ using $I_2$.
Then $C(K_1)$ and $C(K_2)$ end at different output states.

Since $i$ must be inactive for $J$, there are only two ways $J$ could fail to satisfy the claim: either $J$ has a black pebble $[i, a']$ where $a \neq a'$, or $J$ has no black pebble on $i$.
In the first case, $f_j(a, b)$ is a non-thrifty query for $K_1$ and $K_2$, so these inputs have the same correct root values, a contradiction.
In the second case, $i$ is inactive and its parent $j$ is active (because it is active for $I$). If $|Range^J_\gamma(i)| > 1$, then by Proposition \ref{prop:parent_gap}, $j$ must be completely queried at $\gamma$ for $J$.
By Proposition \ref{prop:same_gaps}, $J$ cannot have queried $f_j(a,b)$ before $\gamma$, so $f_j(a,b)$ is a non-thrifty query for $J$ (and hence $K_1$ and $K_2$), leading to the same contradiction.
Finally, if $|Range^J_\gamma(i)| = 1$ but $i$ \emph{isn't} black pebbled for $J$ at $\gamma$, there must be a prior state on $C_0(J, \gamma)$ where a black pebble was removed from $i$.
At this state, $j$ was active (because it is still active at $\gamma$), and so it must have been completely queried to cause the black pebble to be removed. This once again leads to the same contradiction.
\end{proof}

The remaining two propositions deal with how $\gamma$ encodes the grey pebbles.
\begin{proposition}
\label{prop:simul_satisfy}
Let $\gamma$ be a state on two computation paths $C(I)$ and $C(J)$ on a deterministic, read-once branching program solving $TEP^h_2(k)$.
Let $C$ be a collection of active nodes such that none of them are ancestors of each other.
For each $i \in C$, let $a_i \in [k]$.
Suppose that every active node has a gap for both $I$ and $J$ at $\gamma$.
Then there exist two inputs $I'$ and $J'$ agreeing with $I$ and $J$ before $\gamma$, respectively, and agreeing with each other after $\gamma$, such that for all $i \in C$, $v_i^{I'} = v_i^{J'} = a_i$.
\end{proposition}
\begin{proof}
As in Proposition \ref{prop:range}(ii), it suffices to show this for one node.
Let $i \in C$.
If $i$ is a leaf, then it has not been queried on $C_0(I,\gamma)$ or $C_0(J,\gamma)$ (otherwise it would be inactive). 
Therefore we can simply choose $C_1(I',\gamma) = c_1(J',\gamma)$ to have that leaf value be $a_i$. 

Now suppose $i$ is an internal node, and suppose $I$ has the gap $f_i(x,y)$.
By Proposition \ref{prop:same_gaps}, $J$ did not make this query before $\gamma$.
We claim that $x \in Range^J_\gamma(2i)$ and $y \in Range^J_\gamma(2i + 1)$.
This is clear if the children of $i$ are both active, and hence have a gap: in this case, $Range^J_\gamma(2i) = Range^J_\gamma(2i + 1) = [k]$.
On the other hand, by Proposition \ref{prop:parent_gap}, because $i$ is not completely queried for $J$, the only way its child could be inactive is if its correct value has been determined - i.e., if it has a black pebble.
But then by Proposition \ref{prop:same_black}, $I$ has the same black pebble, and hence $Range^I_\gamma(i) = Range^J_\gamma(i)$.
This completes the proof of the claim.
By induction, there exist inputs $I'$ and $J'$ that both have $x,y$ as the correct values of the children of $i$, and so setting $f_i^{I'}(x,y) = f_i^{J'}(x,y) = a_i$ gets the desired result.
\end{proof}

\begin{proposition}
\label{prop:same_grey}
Let $\gamma$ be a state on two computation paths $C(I)$ and $C(J)$ on a deterministic, read-once branching program solving $TEP^h_2(k)$.
Suppose that every active node has a gap for both $I$ and $J$ at $\gamma$.
Suppose $I$ has a grey pebble $[2i, a] \wedge [2i+1, b] \imp [i, c]$ at $\gamma$.
Then $J$ has a grey pebble $[2i, a] \wedge [2i + 1, b] \imp [i, d]$ and $[i, c] \approx [i, d]$ with respect to both $I$ and $J$ at $\gamma$.
\end{proposition}
\begin{proof}
The first thing to show is that $J$ must have a grey pebble $[2i, a] \wedge [2i + 1, b] \imp [i, d]$ for some $d$.
Since $C_0(I,\gamma)$ has queried $f_i(a, b)$, $J$ cannot have this function value as a gap at $\gamma$.
Since $i$ is active for $I$ it must be active for $J$; therefore the only way $J$ could \emph{not} have such a grey pebble is if (without loss of generality) $a \notin Range^J_\gamma(2i)$.
Since all active nodes have a gap and hence have a $Range$ of $[k]$, this means that node $2i$ must be inactive, and have black pebbles (since their parent $i$ is active and not completely queried).
But by Proposition \ref{prop:same_black}, $I$ must have the same black pebbles, and hence $Range^I_\gamma(2i) = Range^J_\gamma(2i)$.

So $J$ has a grey pebble $[2i, a] \wedge [2i + 1, b] \imp [i, d]$.
First suppose $i$ is the root and $c \neq d$.
By Proposition~\ref{prop:simul_satisfy}, we can find inputs $I'$ and $J'$ that agree with $I$ and $J$ before $\gamma$, respectively, and agree with each other after $\gamma$, and have $v^{I'}_1 = c$ and $v^{J'}_1 = d$, a contradiction.
For non-root $i$, we need to use the fact that $[i, c]$ and $[i, d]$ are not equivalent for $I$.
In this case, first pick $I'$ and $J'$ as in Proposition \ref{prop:simul_satisfy} so that $v_{2i}^{I'} = v_{2i}^{J'} = a$ and $v_{2i+1}^{I'} = v_{2i+1}^{J'} = b$.
Since $I'$ agrees with $I$ before $\gamma$, $[i, c] \not \approx^{I'}_\gamma [i, d]$.
Then by Proposition \ref{prop:equiv_2}, we can also choose for $I'$ correct values for the siblings of the ancestors of $i$ so that different correct values $c$ and $d$ for node $i$ result in different correct root values.
By definition, the siblings of the ancestors of $i$ together with the children of $i$ also satisfy the hypothesis of Proposition \ref{prop:simul_satisfy}. 
So if we choose $J'$ to also have these correct values for the siblings of the ancestors of $i$, $I'$ and $J'$ will have different correct root values, yet agree after $\gamma$.
\end{proof}

Putting the black pebble and grey pebble results together yields the following result.
\begin{lemma}
\label{lem:ident_config}
Let $\gamma$ be a state on two computation paths $C(I)$ and $C(J)$ on a deterministic, read-once branching program solving $TEP^h_2(k)$.
If every active node at $\gamma$ has a gap for both $I$ and $J$, then the pebble configurations associated with $\gamma$ for $I$ and $J$ are identical up to equivalence for grey pebbles.
\end{lemma}

\subsection{The Lower Bound}
We use the results of the previous subsection to argue that there must be a large number of states, for large enough $k$.
The intuition is the following.
Even though the pebbling rules are now more complicated, the new rules only apply when at least $k$ grey pebbles have been put onto a node. 
The most efficient way of shrinking a $Range$ is still to use thrifty queries, and this corresponds to valid black pebbling moves.

\begin{definition}[Relevant Query] Fix a computation path $C(I)$ and state $\gamma$ on it. We say that $\gamma$ makes a \emph{relevant query} $f_i(x,y)$ if $x \in Range_\gamma(2i)$ and $y \in Range_\gamma(2i+1)$.
A query $f_i(x,y)$ on $C_0(I, \gamma)$ is \emph{relevant} at state $\gamma$ if $x \in Range_\gamma(2i)$ and $y \in Range_\gamma(2i+1)$. Note that this definition only applies for internal nodes.
\end{definition}

\begin{definition}[Efficient] A computation path $C(I)$ is \emph{efficient} if for each state $\gamma$ on the path, there is no active node $i$ for which $k-1$ relevant queries have been made on $C_0(I,\gamma)$.
\end{definition}

\begin{proposition}
\label{prop:efficient_prop}
Let $\gamma$ be a state on an \emph{efficient} computation path $C(I)$ on a deterministic, read-once branching program solving $TEP^h_2(k)$.
Let $i$ be a node. Then the following hold:
\begin{enumerate}
\item[(i)]
If $i$ is an internal node and either $Range_\gamma(2i) = [k]$ or $Range_\gamma(2i+1) = [k]$, then $i$ is not completely queried.
\item[(ii)] Either $|Range_\gamma(i)| = k$ or $|Range_\gamma(i)| = 1$.
\item[(iii)] Node $i$ is inactive if and only if $|Range_\gamma(i)| = 1$.
\end{enumerate}
\end{proposition}
\begin{proof}
For Claim (i), suppose without loss of generality that there is an internal node $i$ with $Range(2i) = [k]$.
If $i$ is completely queried, let $\delta$ be the last state on $C(I)$ before or equal to $\gamma$ where $i$ is not completely queried.
At $\delta$, $Range_\delta(2i) = [k]$ because $Range$ never grows. Also, $Range_\delta(i) = [k]$ because $i$ has a gap, hence $i$ is active at $\delta$.
Then since $\delta$ is the last state where $i$ has a gap, there must be at least $k-1$ relevant queries to $i$, contradicting efficiency.

Claim (ii) is certainly true for leaves.
Suppose $i$ is an internal node. By Claim (i), if one of its children has a $Range$ of $[k]$, then $Range_\gamma(i) = [k]$.
Otherwise, by induction $|Range_\gamma(2i)| = |Range_\gamma(2i+1)| = 1$, and then $C_0(I, \gamma)$ has either made the thrifty query to $i$ or not.
These two cases correspond to $|Range_\gamma(i)| = 1$ and $|Range_\gamma(i)| = k$, respectively.

For Claim (iii), the backwards direction follows immediately from the definition. 
By Claim (ii), we only need to consider the case $Range_\gamma(i) = [k]$. 
If $i$ is the root we are done because each number is its own equivalence class for the root.
Otherwise, by Claim (i), the parent of $i$ is active and not completely queried.
Then $i$ must be active by Proposition \ref{prop:parent_gap}.
\end{proof}

Next we establish that the pebbling sequence constructed by our pebbling algorithm is essentially a black pebbling. 
\begin{proposition}
\label{prop:end_black}
Let $C(I)$ be a computation path on a deterministic, read-once branching program solving $TEP^h_2(k)$.
The pebbling configuration corresponding to the output state of $C(I)$ has a black pebble on the root.
\end{proposition}
\begin{proof}
Observe that black pebbles are never removed from the root.
Therefore the root could only be not pebbled if it were still active; but by Proposition \ref{prop:equiv_2}, it cannot be active.
\end{proof}

\begin{proposition}
\label{prop:efficient_black_move}
Let $C(I)$ be an efficient computation path on a deterministic, read-once branching program solving $TEP^h_2(k)$.
Then whenever an internal node is black-pebbled, its children have black pebbles on them.
\end{proposition}
\begin{proof}
Choose a point in the algorithm when an internal node $i$ gets black pebbled.
By Proposition \ref{prop:efficient_prop}, this occurs at the first state $\gamma$ where $i$ is inactive.
At this point, its children must be inactive, so $|Range_\gamma(2i)| = |Range_\gamma(2i+1)| = 1$.
Then they must have been first black pebbled at or before $\gamma$.
If they are first pebbled at $\gamma$ as well, then they are still pebbled when $i$ is black pebbled (because the nodes are processed bottom-up).
Now without loss of generality suppose node $2i$ was black pebbled at a state $\delta$ before $\gamma$.
Node $i$ could not be inactive at $\delta$, since it becomes inactive only at $\gamma$.
Also, $i$ cannot be completely queried: if $|Range_\delta(2i+1)| = [k]$, because of Proposition \ref{prop:efficient_prop}; if $|Range_\delta(2i+1)| = 1$, then $i$ would be also be black pebbled at $\delta$.
Then this pebble on $2i$ couldn't have been removed before $\gamma$.
\end{proof}

These lead to the following lemma.

\begin{lemma}
\label{lem:valid}
Let $C(I)$ be an efficient computation path on a deterministic, read-once branching program solving $TEP^h_2(k)$.
Then the \emph{black} pebbling sequence obtained by removing all grey pebbles from the configurations returned by the pebbling algorithm run on $C(I)$ is a valid black pebbling of $T^h_2$.
\end{lemma}

Unfortunately, there is not a direct correspondence between black pebble moves and states - it could be that the intermediate configurations associated with a single state contain multiple black sliding moves.
However, this is only the case if the corresponding grey pebbles are present in the configurations.

\begin{lemma}
\label{lem:bound_1}
Let $C(I)$ be an efficient computation path on a deterministic, read-once branching program solving $TEP^h_2(k)$.
Then there is a state on $C(I)$ whose associated pebbling configuration has at least $h$ pebbles.
\end{lemma}
\begin{proof}
By Lemma \ref{lem:valid}, the sequence of black pebble configurations is a valid pebbling, so we can apply the black pebbling lower bound of Theorem \ref{thm:black_num} to find a configuration with $h$ black pebbles on it.
Let $C$ be the first such configuration, and let $C_1$ be the first configuration equal to or after $C$ that is associated with a state.
Suppose $C_1$ does \emph{not} have $h$ pebbles.
In particular, $C_1$ has fewer black pebbles than $C$, and so between $C$ and $C_1$ at least one black pebble was removed.
Note that black pebbles are only removed during black sliding moves, which only take place if there is a grey pebble on the node that is black pebbled as a result of the sliding move. 
Let $C_0$ be the last configuration before $C$ that is associated with a state.
$C_0$ has at least $h-1$ black pebbles, and since a black sliding move occurs between $C_0$ and $C_1$, $C_0$ has at least one grey pebble.
\end{proof}

\begin{lemma}
Let $k \geq h + 1$. Then every input $I$ has a state on $C(I)$ where the associated pebbling configuration has $h$ pebbles (black or grey).
\end{lemma}
\begin{proof}
If $C(I)$ is efficient, then by Lemma \ref{lem:bound_1}, $C(I)$ has a state whose associated pebbling configuration has at least $h$ pebbles.
Since the pebble number from one state to the next always increases by at most 1, this implies there is a state on the path with exactly $h$ pebbles.
If $C(I)$ is not efficient, then there is a state $\gamma$ and active node $i$ with at least $k-1 \geq h$ relevant queries made to $f_i$ before $\gamma$.
These $k-1$ queries are each represented by a different grey pebble on node $i$ at $\gamma$, and the claim follows.
\end{proof}

Finally, we can get a lower bound on the total number of states. We define the \emph{supercritical state} of an input as the first state on the computation path whose associated pebble configuration has at least $h$ pebbles, black or grey. We are now ready to prove Theorem \ref{thm:ro_det}.

\begin{proof}[Proof of Theorem \ref{thm:ro_det}]
Let $f$ be a function mapping each input $I$ to its \emph{supercritical state} $\gamma$.
Each black pebble specifies the correct value of a particular node.
Each grey pebble specifies one value of the input, which could be correct or not correct.
Moreover, for an active node $i$ at $\gamma$, each variable $[i, a]$ must be in its own equivalence class, because no entire row or column has been queried yet.
Therefore applying Lemma \ref{lem:ident_config} implies that the $h$ pebbles specify precisely $h$ variables in the input, and the claim follows.
\end{proof}

\section{Bitwise-Independent Thrifty}
Recently, Komarath and Sarma \cite{ks13} proved lower bounds for non-deterministic thrifty branching programs with a new semantic restriction, which was the first non-trivial bound for any family of non-deterministic branching programs that applied to arbitrary $h$. They also introduced a notion of ``state pebble values" which elegantly captures the pebble metaphor in a novel manner, differing from previous pebbling arguments because these values are intrinsic to states and do \emph{not} depend on the state sequence of a computation path. In this section we present their ideas, giving a simplified proof of their main lower bound. In the following section, we extend this proof to non-deterministic \emph{syntactic read-once} branching programs.

\subsection{Definitions}
Fix any branching program solving $TEP^h_2(k)$. For each state $\gamma$ and node $i$, we define the two following sets $R_\gamma(i), A_\gamma(i) \subseteq [k]$:
\begin{align*}
R_\gamma(i) &= \{v_i^I \mid \text{some computation path $C(I)$ reaches $\gamma$}\} \\
A_\gamma(i) &= \{v_i^I \mid \text{some \emph{complete} computation path $C(I)$ reaches $\gamma$}\}
\end{align*}
Intuitively, $R_\gamma(i)$ is constructed by taking all inputs which can reach $\gamma$ and projecting them down onto their $v_i$-coordinate, and similarly for $A_\gamma(i)$ with the additional restriction that the inputs reach an output state from $\gamma$. Note that these are closely related to the $Range_\gamma(i)$ of the previous section.

\begin{definition}[State Pebble Values]
For each state $\gamma$ and node $i$, we define the \emph{state black and white pebble values} $b_\gamma(i)$ and $w_\gamma(i)$ according to the following formulas.
$$
b_\gamma(i) = \log_k \left( \frac{k}{|R_\gamma(i)|} \right) \quad \text{ and } \quad w_\gamma(i) = \log_k \left( \frac{|R_\gamma(i)|}{|A_\gamma(i)|} \right)
$$
The \emph{total state pebble value} of a node $i$ at $\gamma$ is $p_\gamma(i) = b_\gamma(i) + w_\gamma(i)$, and the \emph{total state pebble value} of a state $\gamma$ is $p_\gamma = \sum_i p_\gamma(i)$.
\end{definition}
Here are some basic properties of these definitions, which confirm that these state pebble values are at least somewhat consistent with our intuition regarding pebbling arguments for branching programs.
\begin{proposition} 
Let $B$ be a minimal non-deterministic branching program solving $TEP^h_2(k)$. For any state $\gamma$ and node $i$,
\begin{enumerate}
\item[(i)] $\emptyset \neq A_\gamma(i) \subseteq R_\gamma(i) \subseteq [k]$
\item[(ii)] $0 \leq b_\gamma(i), w_\gamma(i) \leq 1$
\item[(iii)] $p_\gamma(i) = \log_k \left( \frac{k}{|A_\gamma(i)|} \right) \leq 1$
\item[(iv)] If $B$ is deterministic, then $A_\gamma(i) = R_\gamma(i)$, hence $w_\gamma(i) = 0$.
\end{enumerate}
\end{proposition}

\begin{proposition}
\label{prop:bi_count}
Let $\gamma$ be a state on a branching program solving $TEP^h_2(k)$. Then at most $1/k^{p_\gamma}$ of the inputs have a complete computation path through $\gamma$.
\end{proposition}
\begin{proof}
Let $m = 2^h - 1 + (k^2 - 1)(2^{h-1} - 1)$ be the number of $k$-values required to specify the input to $TEP^h_2(k)$.
Note that by definition, $$k^{p_\gamma} = \frac{k^{2^h - 1}}{\prod_i |A_\gamma(i)|}.$$ 
There are at most $\displaystyle{\prod_i |A_\gamma(i)|}$ combinations of correct node values for the inputs having a complete computation path through $\gamma$. Each of these combinations correspond to at most $k^{m - (2^h - 1)}$ distinct inputs (one for each combination of non-thrifty function values), for a total of at most $\left(\prod_i |A_\gamma(i)|\right)k^{m - (2^h - 1)}$ inputs with a complete computation path through $\gamma$. Then a simple calculation shows that 
$$\left(\prod_i |A_\gamma(i)|\right)k^{m - (2^h - 1)} = k^m \cdot \frac{\prod_{i=1}^{2^h - 1} |A_\gamma(i)|}{k^{2^h - 1}} = k^m \cdot k^{-p_\gamma} = k^{m - p_\gamma}.$$
\end{proof}

The above definitions and results are well-defined for any branching program. However, they seem to be most meaningful for a small subset of thrifty branching programs. One shortcoming of the standard pebbling argument is that pebbles can generally be moved independently of each other, except for the parent-child conditions on black placing and white removing. However, general branching programs are free to treat (correct) node values in aggregate rather than separately, and thus states may encode \emph{correlations} between node values. One way to view the two read-once lower bounds of Chapters 4 and 5 is getting around this problem by restricting the queries that the branching program can make, so that correlations cannot be ``used" by the branching program to save space. This issue becomes even more severe for non-deterministic branching programs because now it is possible for states to ``guess" correlations between nodes. Because pebbling does not seem to capture these correlations easily, and our main goal here is to explore the power and limits of the pebbling argument, we define two semantic restrictions on branching programs motivated by ruling out these correlations. While the first is a more natural restriction for the Tree Evaluation Problem, the second is stronger and seems necessary to prove the desired lower bound.

\begin{definition}[Node-independence]
A branching program solving $TEP^h_2(k)$ is \emph{node-independent} if for all states $\gamma$ and inputs $I$, the following conditions hold. 
\begin{enumerate}
\item[(1)] $I$ reaches $\gamma$ if and only if for all nodes $i$, $v_i^I \in R_\gamma(i)$
\item[(2)] $I$ has a complete computation path through $\gamma$ if and only if for all nodes $i$, $v_i^I \in A_\gamma(i)$
\end{enumerate}
\end{definition}

Note that the forward direction in both conditions follows directly from the previous definitions, so it is only the backwards direction that makes this a strong restriction. Geometrically, the inputs that reach (complete through) $\gamma$ form a combinatorial rectangle, the direct product of the $R_\gamma(i)$'s ($A_\gamma(i)$'s). 

The next definition goes one step further, and says that the states may not even remember correlations between the \emph{bits} of these values.  
We note that Komarath and Sarma actually use a slightly more general restriction which allows for ``encodings" $\varphi: [k] \to \{0,1\}^{\ceil{\log_2 k}}$ rather than the standard binary representation, but for the purpose of this paper we use this simpler version, for which the same analysis applies and is just as illuminating. Everything we state here generalizes easily to arbitrary encoding functions.

\begin{definition}[Bitwise-independence]
Suppose $k$ is a power of 2. A branching program solving $TEP^h_2(k)$ is \emph{bitwise-independent} if for each state $\gamma$, the following conditions hold.
\begin{enumerate}
\item[(1)] There exist sets $R_\gamma(i, l) \subseteq \{0,1\}$ such that $I$ reaches $\gamma$ if and only if for every node $i$ and bit position $1 \leq l \leq \log_2 k$, the $l$-th bit of $v^I_i$ is in $R_\gamma(i,l)$.
\item[(2)] The analogous statement for sets $A_\gamma(i,l) \subseteq \{0,1\}$ and inputs which have a complete computation path through $\gamma$.
\end{enumerate}
\end{definition}

Clearly, every bitwise-independent branching program is also node-independent.

\subsection{Pebbling Sequence}
Let $C(I)$ be a complete computation path. We will associate with $C(I)$ a pebbling sequence using the state black and white pebble values along the path. Our presentation here is a simplification of the original analysis in \cite{ks13}, which defined a series of critical states and pebbling sequence separate from (but related to) the $b_\gamma(i)$ and $w_\gamma(i)$. As we observed above, the state pebble values fall in the correct ranges, so the key thing to prove is that this sequence of pebble values follows valid pebbling rules.

\begin{proposition} 
\label{prop:bi_thrifty}
Let $\gamma$ be a state on a \emph{thrifty} branching program solving $TEP^h_2(k)$. If $\gamma$ queries internal node $i$, then $|A_\gamma(2i)| = |A_\gamma(2i + 1)| = 1$, and hence $p_\gamma(2i) = p_\gamma(2i+1) = 1$.
\end{proposition}

\begin{proposition}
\label{prop:bi_black}
Let $C(I)$ be a complete computation path on a \emph{node-independent} branching program solving $TEP^h_2(k)$, and let $\gamma, \delta$ be two consecutive states on $C(I)$. If $\gamma$ doesn't make the thrifty query to node $i$, then $b_\gamma(i) \geq b_\delta(i)$.
\end{proposition}
\begin{proof}
It suffices to show that $R_\gamma(i) \subseteq R_\delta(i)$. Let $a \in R_\gamma(i)$. By node-independence, we may choose non-thrifty values for $I$ so that the parent $j$ of $i$ is a constant function. Let $I'$ be an input which differs from $I$ only at the thrifty query to node $i$, having $v^{I'}_i = a$. By node-independence $I'$ also reaches $R_\gamma(i)$, and takes the same edges as $I$ out of $\gamma$, and so $I'$ reaches $\delta$ and $a \in R_\delta(i)$.
\end{proof}

\begin{proposition}
\label{prop:bi_white}
Let $C(I)$ be a complete computation path on a \emph{bitwise-independent} branching program solving $TEP^h_2(k)$, and let $\gamma, \delta$ be two consecutive states on $C(I)$. If $\gamma$ doesn't make the thrifty query to node $i$, then $w_\gamma(i) \leq w_\delta(i)$.
\end{proposition}
\begin{proof}
Suppose $\gamma$ doesn't make the thrifty query to $i$. By Proposition \ref{prop:bi_black}, $R_\gamma(i) \subseteq R_\delta(i)$. Therefore the only way the white pebble value could decrease is for $|A_\gamma(i)| < |A_\delta(i)|$. Let $a \in R_\gamma(i) \cap A_\delta(i)$. Again we may assume that the parent of $i$ is a constant function in $I$, and define an input $I'$ identical to $I$ except with the thrifty query to node $i$ having value $a$. By node-independence, $I'$ reaches $\gamma$, and then reaches $\delta$; and by node-independence again, $I'$ has a complete path through $\delta$, and hence also one through $\gamma$. Therefore $A_\delta(i) \backslash A_\gamma(i)$ and $R_\gamma(i)$ must be disjoint.

Here is where we use \emph{bitwise}-independence. If $|A_\delta(i)| > |A_\gamma(i)|$, then bitwise-independence implies that $\frac{|A_\delta(i)|}{|A_\gamma(i)|} = 2^r$ for some $r$. Without loss of generality, assume that the first $r$ bit positions for $v_i$ go from having one choice in $A_\gamma(i)$ to two choices for $R_\delta(i)$; that is, $|A_\gamma(i, l)| = 1$ and $|A_\delta(i, l)| = 2$ for all $1 \leq l \leq r$. Note that $v_i^I \in A_\gamma(i) \cap A_\delta(i)$. For each $1 \leq l \leq r$, define $w_l \in [k]$ to be the number obtained by taking the binary representation of $v_i^I$ and flipping the $l$-th bit. Then $w_l \in A_\delta(i) \backslash A_\gamma(i)$, and hence $w_l \notin R_\gamma(i)$. 
Since $v_i^I \in R_\gamma(i)$, by bitwise-independence $|R_\gamma(i, l)| = 1$. But since $A_\delta(i) \subseteq R_\delta(i)$, $|R_\delta(i, l)| = 2$. Since $R_\gamma(i) \subseteq R_\delta(i)$, we get that $\frac{|R_\delta(i)|}{|R_\gamma(i)|} \geq 2^r = \frac{|A_\delta(i)|}{|A_\gamma(i)|}$, and the claim follows.
\end{proof}

The above three propositions essentially rule out the invalid pebbling moves by regulating the increases in black pebble values and decreases in white pebble values along a computation path. However, simply using the pebble values at each state is not quite a valid pebbling sequence because moves may be skipped; an arbitrary number of black removing or white placing moves can happen between states. 
Thus we have natural pebble configurations obtained from the state pebble values, and it remains to define intermediate pebble moves between states to create a valid pebbling sequence. Suppose $\gamma, \delta$ are two consecutive states on $C(I)$ and $\gamma$ queries node $i$. For each pebble value change $b_\gamma(j) \neq b_\delta(j)$ and $w_\gamma(j) \neq w_\delta(j)$, we will have one valid move which changes the pebble value at $\gamma$ to the value at $\delta$. We perform these moves in the following order:

\begin{enumerate}
\item[(1)] For every node $j$ \emph{except} $i$ or its children, decrease the black pebble value on node $j$.
\item[(2a)] If $w_\gamma(i) > w_\delta(i)$, decrease the white pebble value of node $i$.
\item[(2b)] If $b_\gamma(i) < b_\delta(i)$, increase the black pebble value of node $i$ \emph{while simultaneously} decreasing black pebble values of the children of $i$. 

Otherwise, just decrease the black pebble values on the children of $i$.
\item[(3)] Increase the white pebble values of any other nodes.
\end{enumerate}

Note that pebbles are never added until absolutely necessary. By Propositions \ref{prop:bi_black} and \ref{prop:bi_white}, $i$ is the only node where the black pebble value can increase or the white pebble value decrease, and by Proposition \ref{prop:bi_thrifty}, the children of $i$ are fully pebbled before these moves occur. This pebbling sequence begins with the empty configuration (since $A_\gamma(i) = [k]$ at the start state, for all $i$) and ends with a single black pebble on the root (corresponding to the output state). Adding a final move to remove the black pebble results in a valid pebbling sequence.

\subsection{The Lower Bound}
Now we would like to use the pebble number of the fractional pebbling gamer, whose applicability is an easy consequence of the chosen order of the intermediate pebbling moves.

\begin{proposition}
Let $C(I)$ be a complete computation path on a bitwise-independent, thrifty branching program solving $TEP^h_2(k)$. Then some state $\gamma$ has $p_\gamma \geq \frac{h}{2} + 1$.
\end{proposition}
\begin{proof}
Associate with $C(I)$ the valid pebbling sequence described in the previous subsection. By the fractional pebbling lower bound of Theorem \ref{thm:frac_num}, some configuration on the pebble sequence has at least $\frac{h}{2} + 1$ pebbles; if this configuration corresponds to a particular state, we are done. Otherwise, say this configuration lies strictly between consecutive states $\gamma$ and $\delta$ on $C(I)$. Then it is produced by some intermediate pebbling move. If it is produced by steps (1) or (2a) then $p_\gamma > \frac{h}{2} + 1$, since these moves only decrease the pebble number. If it is produced by steps (2b) or (3), then $p_\delta \geq \frac{h}{2} + 1$, since any subsequent moves only increase the pebble number.
\end{proof}

Thus we have shown that some state has many pebbles, while Proposition \ref{prop:bi_count} showed that pebbles are meaningful. The following theorem combines these in the standard way to achieve the lower bound.

\begin{theorem}[\cite{ks13}]
\label{thm:bit}
Every bitwise-independent, thrifty branching program solving $TEP^h_2(k)$ has at least $k^{h/2 + 1}$ states.
\end{theorem}
\begin{proof}
Map each input $I$ to the first state on (an arbitrarily selected) complete computation path $C(I)$ having a total pebble value of at least $\frac{h}{2} + 1$. By Proposition \ref{prop:bi_count}, at most $1/k^{\frac{h}{2} + 1}$ of the inputs are mapped to the same state.
\end{proof}

\section{Node-Independent Read-Once}
Our final contribution is to combine node-independence from the previous section with the syntactic read-once restriction. The following argument is a simple variant of the arguments which have come before, with only a handful of new technical points required.

\begin{proposition}
\label{prop:niro_mix}
Let $C(I)$ be some complete computation path on a branching program solving $TEP^h_2(k)$.  Let $\gamma$ be some state on $C(I)$ and $i$ some node. If $b_\gamma(i) > 0$, then some state on $C_0(I, \gamma)$ made the thrifty query to node $i$; if $w_\gamma(i) > 0$, then some state on $C_1(I, \gamma)$ will make the thrifty query to node $i$.
\end{proposition}
\begin{proof}
This is a straightforward consequence of the fact that for $b_\gamma(i) > 0$, it must be true that $R_\gamma(i) \neq [k]$, and hence some possible values for $v_i$ must have been ``rejected" by a previous query. Similarly, if $w_\gamma(i) > 0$, then $R_\gamma(i) \neq A_\gamma(i)$, so some of the values reaching $\gamma$ must be rejected before arriving at an output state.
\end{proof}

Thus for read-once branching programs, it cannot be the case that $b_\gamma(i) > 0$ and simultaneously $w_\gamma(i) > 0$ for any state and node. Now we wish to establish the analogues of Propositions \ref{prop:bi_thrifty} and \ref{prop:bi_white} in this setting. Note that Proposition \ref{prop:bi_black} carries over immediately.

\begin{proposition}
\label{prop:niro_white}
Let $C(I)$ be a complete computation path on a node-independent, syntactic read-once branching program solving $TEP^h_2(k)$, and let $\gamma, \delta$ be two consecutive states on $C(I)$. If $\gamma$ doesn't make the thrifty query to node $i$, then $w_\gamma(i) \leq w_\delta(i)$.
\end{proposition}
\begin{proof}
Suppose $w_\gamma(i) > 0$. Then $b_\gamma(i) = 0$, and so $R_\gamma(i) = [k]$. By the same argument as in Proposition \ref{prop:bi_white} (which only used node-independence), $A_\delta(i) \backslash A_\gamma(i) \cap R_\gamma(i) = \emptyset$, hence $A_\delta(i) \subseteq A_\gamma(i)$, and $w_\gamma(i) \leq w_\delta(i)$.
\end{proof}

\begin{proposition}
\label{prop:niro_glue}
Let $C(I)$ be a complete computation path on a node-independent, syntactic read-once branching program solving $TEP^h_2(k)$, and let $\gamma, \delta$ be two consecutive states on $C(I)$.
Suppose $\gamma$ doesn't make the thrifty query to node $i$. 
If $b_\gamma(i) = 0$ then $A_\delta(i) \subseteq R_\gamma(i)$, and if $b_\gamma(i) > 0$ then $R_\gamma(i) \subseteq A_\delta(i)$.
\end{proposition}
\begin{proof}
If $b_\gamma(i) = 0$ then $R_\gamma(i) = [k]$, so $A_\delta(i) \subseteq R_\gamma(i)$. 
Now suppose $b_\gamma(i) > 0$. Then by Proposition \ref{prop:niro_mix}, $w_\gamma(i) = w_\delta(i) = 0$, and so $R_\delta(i) = A_\delta(i)$. But then by Proposition \ref{prop:bi_black}, $R_\gamma(i) \subseteq R_\delta(i) = A_\delta(i)$.
\end{proof}

\begin{proposition}
\label{prop:niro_children}
Let $C(I)$ be a complete computation path on a node-independent, syntactic read-once branching program solving $TEP^h_2(k)$, and let $\gamma, \delta$ be two consecutive states on $C(I)$.
If $b_\gamma(i) < b_\delta(i)$ or $w_\gamma(i) > w_\delta(i)$, then $\gamma$ makes the thrifty query to node $i$. Moreover, if $i$ is an internal node, then for each child $j$ of $i$, either $R_\gamma(j) = \{v^I_j\}$ or $A_\delta(j) = \{v^I_j\}$, and hence either $b_\gamma(j) = 1$ or $w_\delta(i) = 1$.
\end{proposition}
\begin{proof}
The first statement is a direct consequence of Propositions \ref{prop:bi_black} and \ref{prop:niro_white}.
Suppose $i$ is an internal node, and let $j$ be a child of $i$.
First assume that $b_\gamma(i) < b_\delta(i)$, and let $a \in R_\gamma(i) \backslash R_\delta(i)$. 
If there is some $x \in R_\gamma(j)$ distinct from $v^I_j$, then by node-independence there exists an input $J$ reaching $\gamma$ such that $v_i^J = a$ and $v_j^J = x$. But then $\gamma$ queries a non-thrifty value with respect to $J$, so we can choose this $J$ to follow the same edge of out $\gamma$ as $I$, and hence $a \in R_\delta(i)$, a contradiction. Therefore in this case $R_\gamma(j) = \{v^I_j\}$.

Now assume that $b_\gamma(i) \geq b_\delta(i)$ and $w_\gamma(i) > w_\delta(i)$. Then $w_\gamma(i) > 0$, so by Proposition \ref{prop:niro_mix}, $b_\gamma(i) = b_\delta(i) = 0$ and hence $R_\gamma(i) = R_\delta(i) = [k]$. Then $|A_\gamma(i)| < |A_\delta(i)|$, so let $a \in A_\delta(i) \backslash A_\gamma(i)$. 
Let $j$ be a child of $i$. 
If $b_\gamma(j) = 0$ and $A_\delta(j) \neq \{v^I_j\}$ then choose some $x \in A_\delta(j)$ distinct from $v^I_j$. 
If $b_\gamma(j) > 0$ and $R_\gamma(j) \neq \{v^I_j\}$, then choose $x \in R_\gamma(j)$ distinct from $v^I_j$.

In either case, by Proposition \ref{prop:niro_glue}, $x \in R_\gamma(j) \cap A_\delta(j)$, and so there exists an input $J$ which has a complete computation path through $\gamma$ and $\delta$, has $v^J_i = a$ and $v^J_j = x$, noting that $\gamma$ makes a non-thrifty query with respect to $J$. 
Then $a \in A_\gamma(i)$, a contradiction.
\end{proof}

Thus for this family of branching programs, the state pebble values along complete computation paths still correspond to pebble configurations in a valid pebble sequence. The intermediate pebble moves are the same as the bitwise-independent thrifty case, with one exception caused by the additional complexity of Proposition \ref{prop:niro_children}. The new step \textbf{2'} is shown in bold.

\begin{enumerate}
\item[(1)] For every node $j$ \emph{except} $i$ or its children, decrease the black pebble value on node $j$.
\item[\textbf{(2')}] \textbf{Increase the white pebble values of the children of $i$}.
\item[(2a)] If $w_\gamma(i) > w_\delta(i)$, decrease the white pebble value of node $i$.
\item[(2b)] If $b_\gamma(i) < b_\delta(i)$, increase the black pebble value of node $i$ \emph{while simultaneously} decreasing black pebble values of the children of $i$. 

Otherwise, just decrease the black pebble values on the children of $i$.
\item[(3)] Increase the white pebble values of any other nodes.
\end{enumerate}

The extra step is necessary because the thrifty restriction guaranteed that if node $i$ is queried at $\gamma$ then its children are fully pebbled at $\gamma$, but in the node-independent, read-once case it is only guaranteed that a white pebble is placed by the \emph{following} state. It is easy to verify that applying these steps to the entire computation path yields a valid pebbling sequence (again, adding a final move to remove the black pebble on the root). Unfortunately, the extra step can cause an \emph{increase} in pebble value followed by a decrease, so the maximum pebble configuration may be skipped if we consider only the configurations associated with states.

\begin{proposition}
Let $C(I)$ be a complete computation path on a node-independent, syntactic read-once branching program solving $TEP^h_2(k)$. There is some state $\gamma$ on $C(I)$ such that if $i$ is the node that $\gamma$ queries, then $p_\gamma - p_\gamma(2i) - p_\gamma(2i+1) + 2 \geq \frac{h}{2} + 1$.
\end{proposition}
\begin{proof}
Since the associated pebbling sequence is valid, there is some configuration that has at least $\frac{h}{2} + 1$ pebbles. Suppose the first such configuration lies between consecutive states $\gamma$ and $\delta$ on $C(I)$, possibly corresponding to $\gamma$. Note that if it corresponds to $\gamma$, we are done. 

Otherwise, this configuration is produced by some intermediate pebbling move. If this happens during steps (1) or (2a), $p_\gamma \geq \frac{h}{2} + 1$. If it happens during steps (2b) or (3), $p_\delta \geq \frac{h}{2} + 1$. Finally, step (2') increases the white pebble values by exactly $2 - b_\gamma(2i) - w_\gamma(2i) - b_\gamma(2i+1) - w_\gamma(2i+1)$ and this is the first intermediate move after $\gamma$ to increase the pebble number. So if the maximum configuration is produced by step (2'), then $\gamma$ satisfies the claim.
\end{proof}

We will define the \emph{supercritical state} of $I$ as follows. If $C(I)$ has a state with total pebble value at least $\frac{h}{2} + 1$, the supercritical state is the first such state. Otherwise, it is the $\gamma$ from the preceding proposition the \emph{supercritical state} of $I$.
Even at the supercritical state, using the pebble values alone is not enough to give us a lower bound. However, the only missing pebble value at $\gamma$ is on the children of $i$; if these were fully pebbled at $\gamma$, the total pebble value at $\gamma$ would be at least $\frac{h}{2} + 1$. 

However, note that even if a configuration with $\frac{h}{2} + 1$ pebbles is produced by step (2'), if step (2a) does not execute, then $p_\delta \geq \frac{h}{2} + 1$, and we could call this the supercritical state instead. Therefore the only time where there is not state on $C(I)$ with total pebble value at least $\frac{h}{2} + 1$ is when both (2') and (2a) occur between a $\gamma$ and $\delta$; but for (2a) to occur, $\gamma$ must make a thrifty query for the computation path. Knowing that $\gamma$ makes a thrifty query allows us to recover the children's correct values directly, \emph{without the state pebble values}!

\begin{lemma}
Let $\gamma$ be a state in a node-independent, syntactic read-once branching program solving $TEP^h_2(k)$. Then at most $1/k^{\frac{h}{2} + 1}$ inputs have $\gamma$ as their supercritical state.
\end{lemma}
\begin{proof}
Let $i$ be the node queried at $\gamma$. If $p_\gamma \geq \frac{h}{2} + 1$, we are done by Proposition \ref{prop:bi_count}. Otherwise, $p_\gamma - p_\gamma(2i) - p_\gamma(2i+1) + 2 \geq \frac{h}{2} + 1$ and $\gamma$ makes a thrifty query to $i$. At most $1/k^{p_\gamma - p_\gamma(2i) - p_\gamma(2i+1)}$ of the inputs have a complete computation path through $\gamma$ (considering only the constraints on the correct values of the nodes which are not children of $i$), and at most $\frac{1}{k^2}$ of these have $\gamma$ make a thrifty query (considering the correct values of the children of $i$). Combining these observations with the inequality completes the proof.
\end{proof}

With this lemma, proving Theorem \ref{thm:niro} is straightforward, mapping inputs to their supercritical states, as has been done repeatedly before in this paper.

\section{Conclusion}
Though the pebbling argument is but one possible line of attack to achieve strong space lower bounds, it currently stands as the one that has yielded the most general results so far. 
Our main contribution has been to deepen the understanding of this style of argument by employing it in new contexts and showing one possible extension of the pebbling game beyond simply thrifty queries. As long as pebbling remains the optimal strategy for solving the Tree Evaluation Problem, finding clever extensions to pebbling arguments will have a great deal of potential. On the other hand, creating a better algorithm is highly non-trivial; one of the incidental implications of this work is ruling out many of the naive strategies one might try to efficiently solve this problem. One advantage intrinsic to these pebbling arguments is that pebbling games are defined on arbitrary DAGs; indeed, we conjecture that all of the arguments used in this paper carry over to the general DAG Evaluation Problem. The restrictions studied in this paper are rather strong, and the natural goal is to see if the proofs can be modified to apply to broader classes of branching programs. Of course, it is still open if and where the pebbling metaphor breaks down. The problem of ``correlations'' between input values discussed in Section 6 seems to pose a significant challenge to this style of argument; since the central issue lies in determining exactly how much branching program states ``know'' about a particular input, information theoretic arguments may be another promising avenue of attack.

Perhaps the most accessible open problems lie with non-deterministic branching programs: specifically, to prove lower bounds for either non-deterministic thrifty BPs or non-deterministic (syntactic or semantic) read-once BPs. In general, the notion of ``equivalence'' defined in Chapter 5 may be a powerful tool for analysing the behaviour of branching programs. In particular, focusing on inputs which have no equivalent variables seems to be a viable approach -- these inputs seem to be ones for which it is most likely that branching programs can do no better than making thrifty queries.

\section*{Acknowledgements}
I would like to thank Steve Cook and Toni Pitassi for their generous support and fruitful discussions, without which this paper would not have been possible.

\bibliographystyle{plain}
\bibliography{pebbling}

\end{document}